\documentclass[10pt,twocolumn,twoside]{IEEEtran}
\usepackage{amsmath,amsfonts}
\usepackage{algorithmic}
\usepackage{array}
\usepackage{textcomp}
\usepackage{stfloats}
\usepackage{url}
\usepackage{verbatim}
\usepackage{graphicx}
\usepackage{balance}
\usepackage{xcolor}
\usepackage{subfig}
\usepackage{orcidlink}
\usepackage{booktabs} 
\usepackage{hyperref}
\hypersetup{hidelinks,
	colorlinks=true,    
	linkcolor=blue,     
	anchorcolor=blue,   
	citecolor=blue,     
	urlcolor=blue,       
	pdfstartview=Fit,
	breaklinks=true
}
\usepackage{cite}
\usepackage{amsmath,amssymb,amsfonts,amsthm}
\usepackage{algorithm}
\usepackage{flushend}
\newtheorem{theorem}{Theorem}
\newtheorem{lemma}{Lemma}
\newtheorem{ass}{Assumption}

\begin{document}
	\title{Cooperative Switched Formation Control of Autonomous Vehicles: An Event-triggered Approach to Input Saturation and Time-delay Challenges}
	
	\author{Ziming Wang\orcidlink{0000-0001-7000-9578}, Guanxuan Jiang\orcidlink{0009-0001-3686-6266}, Yihuai Zhang\orcidlink{0000-0002-7363-7796}, Karl H. Johansson\orcidlink{0000-0001-9940-5929} \IEEEmembership{Fellow,~IEEE} \\and Apostolos I. Rikos\orcidlink{0000-0002-8737-1984} \IEEEmembership{Member,~IEEE}


            \thanks{Corresponding author: Ziming Wang.}
            \thanks{Ziming Wang and Apostolos I. Rikos are with the Artificial Intelligence Thrust, The Hong Kong University of Science and Technology (Guangzhou), Guangzhou, China. Apostolos I. Rikos is also affiliated with the Department of Computer Science and Engineering, The Hong Kong University of Science and Technology, Clear Water Bay, Hong Kong. E-mails:{\tt~zwang216@connect.hkust-gz.edu.cn}, {\tt~apostolosr@hkust-gz.edu.cn}.}
            \thanks{Guanxuan Jiang is with the Computational Media and Arts Thrust, The Hong Kong University of Science and Technology (Guangzhou), Guangzhou, China. E-mail:{\tt~gjiang240@connect.hkust-gz.edu.cn}.}
            \thanks{Yihuai Zhang is with the Intelligent Transportation Thrust, The Hong Kong University of Science and Technology (Guangzhou), Guangzhou, China. E-mail:{\tt~yzhang169@connect.hkust-gz.edu.cn}.}
            \thanks{Karl H. Johansson is with the Division of Decision and Control Systems, KTH Royal Institute of Technology, SE-100 44 Stockholm, Sweden. He is also affiliated with Digital Futures. E-mail:{\tt~kallej@kth.se}.}
            \thanks{The work of A.I.R. was supported by the Guangdong Provincial Project (Grant No. 2024QN11G109).} 
	}
	\maketitle
    
	\begin{abstract}
  This paper presents a collaborative adaptive formation control framework for autonomous vehicles (AVs), that explicitly handles system uncertainties, input saturation, and communication delays. 
  To overcome the inherent physical torque limits of steering and braking actuators, an input saturation compensation mechanism is introduced to render nonlinearities tractable and improve control reliability. 
  Additionally, a delay-compensating auxiliary system is designed to mitigate the effects of communication delays and reduce tracking errors. 
  Our framework incorporates a dynamic-threshold event-triggered control (ETC) strategy to optimize resource usage. 
  Additionally, uncertainty observers and symmetric barrier Lyapunov functions are developed to ensure robust and safe formation maneuvers. 
  Finally, the effectiveness of the proposed approach is validated through numerical simulations of vehicle formations, complemented by a 3D visualization video demonstrating the dynamic fleet reconfiguration process. 
	\end{abstract}
	

	\section{Introduction}
    
	With the rapid growth of transportation demand and the increasing pressure to improve energy efficiency and traffic safety, intelligent transportation systems (ITSs) have received sustained attention from both academia and industry~\cite{intro1,intro2.1,intro2,intro3,intro4}. As an important component of ITSs, vehicular platoon control provides a promising framework for coordinated multi-vehicle operation in complex traffic environments. Through close cooperation among connected vehicles, platoon control can effectively improve highway utilization~\cite{intro1}, lower fuel consumption~\cite{intro2}, and strengthen operational safety~\cite{intro3}. The earliest large-scale implementation of this concept can be found in the California PATH program in the mid-1980s~\cite{intro4}. Since then, platoon control has gradually evolved into a major research topic in automated and connected transportation systems. Its fundamental objective is to maintain coordinated vehicle motion with desired velocity profiles and safe spacing policies, thereby supporting efficient and reliable traffic flow~\cite{intro5,intro6,intro7}.

    Although, most existing platoon control studies are primarily concerned with longitudinal regulation in fixed single-lane configurations \cite{2018, 2019, 2020}, vehicular formation control~\cite{2024, 2025} incorporates both longitudinal and lateral dynamics, thereby providing greater flexibility for handling obstacles, traffic cut-ins, and multi-lane driving scenarios. 
    Therefore, unlike platoon control, vehicular formation control is more valuable because AVs in real-world scenarios change lanes, inevitably requiring coupled lateral and longitudinal control. 
    Motivated by these considerations, this paper investigates the adaptive switched formation control problem for AVs under an intricate environment.

    \subsection{Related Works}

    Most existing vehicular formation control methods rely on continuous-time or periodic implementations, where control signals are updated continuously or at static sampling instants~\cite{intro3,intro4,intro5,intro6}. 
    Such mechanisms may cause unnecessary consumption of communication and computational resources in autonomous vehicular fleets. To address this issue, event-triggered control (ETC) has been introduced as an effective resource-aware strategy, in which controller updates are executed only when a prescribed triggering condition is satisfied. 
    In the vehicular formation control area, several ETC-related results have been reported. For instance, \cite{2019} introduced ETC to reduce inter-vehicle communication frequency, but this inevitably led to communication delays, which were incorporated into the controller. 
    The work in \cite{intro8} proposed an ETC scheme to determine whether sampled data packets should be released to the vehicular ad hoc network for inter-vehicle cooperation, while \cite{intro9} developed a distributed event-triggered platooning method under two distinct triggering conditions. 
    Additionally, \cite{2026} proposed a novel event-triggered sliding mode control algorithm that systematically tackles key challenges, including external disturbances, complex vehicle dynamics, and communication efficiency. 
    External disturbances represent common sources of uncertainty encountered during road traffic operations, such as wind forces, road irregularities, and load variations.
    Motivated by these studies, this paper introduces a dynamic triggering-condition-based ETC strategy to reduce control resource consumption and improve implementation efficiency. 

In practical vehicular systems, time delays arising from sensors, driveline systems, and actuators are unavoidable due to environmental variability and real-time implementation requirements. 
Time delays encompass perception and computation latency as well as actuator response lag, all of which must be explicitly modeled for a realistic controller design.
Existing studies \cite{intro10,intro11,intro12,intro13} have shown that such delays may degrade platoon stability, enlarge tracking errors, and even threaten driving safety, especially in communication-based coordinated control scenarios. To address this issue, this paper develops a delay-compensation auxiliary system to estimate and counteract delay-induced effects, thereby reducing tracking errors and improving system robustness and responsiveness in dynamic environments.

\begin{table*}[!t]
		\caption{Comparison among different vehicular formation control research}
        \label{comp}
		\centering
		\begin{tabular}{cccccc}
			\toprule[2pt] 
			\textbf{Reference} & \textbf{Control direction} & \textbf{Neural networks} & \textbf{Event-triggered control} & \textbf{Input saturation} & \textbf{Time delay}  \\
			\midrule 
                Ge and Han 2017\cite{2017}    & Longitudinal \& Lateral & \text{ } & \checkmark   \\
                Ghaedsharaf et al., 2018\cite{2018}    & Longitudinal   & \text{ } & \text{ } & \text{ } & \checkmark \\
                Keijzer and Ferrari 2019\cite{2019} & Longitudinal & \text{ } & \checkmark    \\
                Khalifa et al., 2020\cite{2020}   & Longitudinal & \checkmark & \text{ } & \text{ } & \checkmark \\
                Chen et al., 2021\cite{2021}   & Longitudinal  & \text{ } & \text{ } & \checkmark  & \text{ } \\
                Liu et al., 2022\cite{2022}     & Longitudinal & \checkmark \\
                Zhai et al., 2023\cite{2023} & Longitudinal & \text{ } & \text{ } & \checkmark  & \checkmark    \\
			    Park and Yoo 2024\cite{2024} & Longitudinal \& Lateral & \text{ } & \text{ } & \checkmark  & \checkmark \\
                Wang et al., 2025\cite{2025} & Longitudinal \& Lateral  & \checkmark  & \checkmark \\
                Li et al., 2026\cite{2026} & Longitudinal & \checkmark & \checkmark \\
                This paper & Longitudinal \& Lateral & \checkmark & \checkmark  & \checkmark & \checkmark \\
			\bottomrule[2pt] 
		\end{tabular}
	\end{table*}

Input saturation is another important issue in controller design, since excessive control inputs may drive actuators beyond their physical limits and consequently degrade system performance. 
Input saturation reflects the physical limitations of throttle, braking, and steering actuators, while state constraints correspond to speed limits imposed by traffic regulations.
In control theory, this problem has been extensively studied in \cite{intro14, intro15, intro16}. 
For example, \cite{intro14} investigated the tracking control problem for a class of multi-input multi-output systems subject to input saturation, while \cite{intro15} addressed adaptive fuzzy state-feedback control for a class of single-input single-output nonlinear systems in nonstrict-feedback form. 
In \cite{intro16}, input saturation is considered in adaptive fuzzy control of an uncertain nonlinear system, which makes the nonlinear system more applicable in fuzzy related practice. 
Motivated by these studies, input saturation is explicitly taken into account in this paper to improve the practical feasibility and reliability of the proposed control framework. 



Existing works have only considered the influence of one or two conditions among neural networks, event-triggered control, input saturation, and time delays on the control process, and are primarily purely theoretical research with limited application scenarios \cite{2018,2019,2020,2024,2025,2026,2017,2021,2022,2023}.
This also is shown in Table~\ref{comp} where we provide a comparative summary of different research on vehicular formation control.
Therefore, most existing studies do not simultaneously address input saturation, time delays, and external disturbances in a unified formation control framework, nor do they provide explicit physical justification for all these constraints together.

\subsection{Main Contributions} 


Motivated by the limitations of the existing literature, we present a dynamic-threshold event-triggered adaptive formation control framework for AVs, developed using the backstepping method, RBF neural networks (NNs), and Lyapunov-based stability theory. To the best of the authors' knowledge, the proposed framework is the first to simultaneously address uncertainties, input saturation, and communication time delays. The main contributions of this paper are summarized as follows: 
\\ \noindent
\textbf{A.} We propose a collaborative adaptive formation control framework that simultaneously handles input saturation, communication delays, and external uncertainties within a unified backstepping-based design. 
We rigorously prove that all closed-loop signals are uniformly ultimately bounded (UUB) and that Zeno behavior is strictly excluded (Theorem~\ref{main_theorem_conv}).
\\ \noindent
\textbf{B.} We demonstrate the framework on a complex formation scenario featuring both variations in fleet velocity and formation pattern changes, which has received limited attention in prior ETC-based formation control literature. 

Compared with traditional platoon control studies \cite{contr1,contr2,contr3}, which mainly focus on longitudinal control, this paper integrates longitudinal and lateral dynamics within a unified framework, thereby achieving coordinated two-dimensional control. 
Moreover, distinct from static-threshold ETC methods \cite{contr4,contr5,contr6}, the ETC strategy developed in this paper adopts a dynamic triggering condition that depends on the current controller value, allowing control resources to be allocated in a more flexible and efficient manner. 
	
The remainder of this paper is organized as follows. Section~\ref{sec2} provides the problem formulation, including the system model and the ETC strategy. Section~\ref{sec3} details the proposed controller design, while Section~\ref{sec4} establishes the Lyapunov-based stability analysis. Section~\ref{sec5} verifies the effectiveness of the adaptive control algorithm through simulations of a switched formation scenario, and Section~\ref{sec6} concludes the paper.

    \begin{figure*}[ht!]
            \centering
             \includegraphics[width=1\textwidth, height=0.25\textheight]{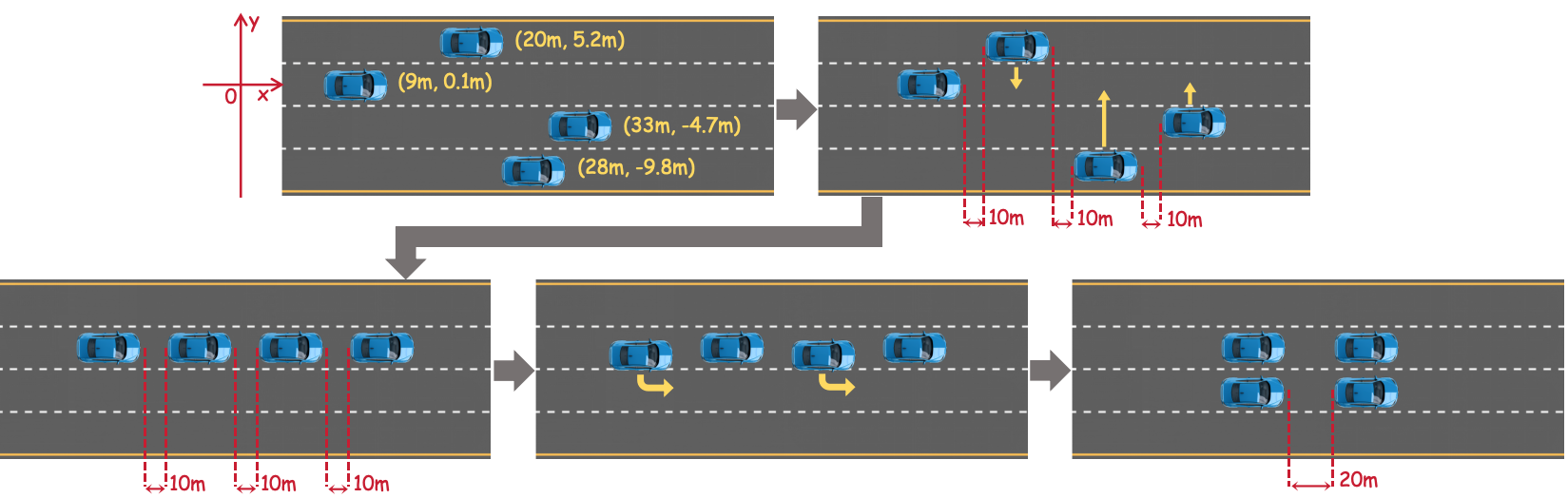}
             \caption{Fleet transition process in cooperative switched formation control of AVs.}
             \label{running}
         \end{figure*}
                \begin{figure}[ht!]
                    \centering
                     \includegraphics[width=0.5\textwidth, height=0.17\textheight]{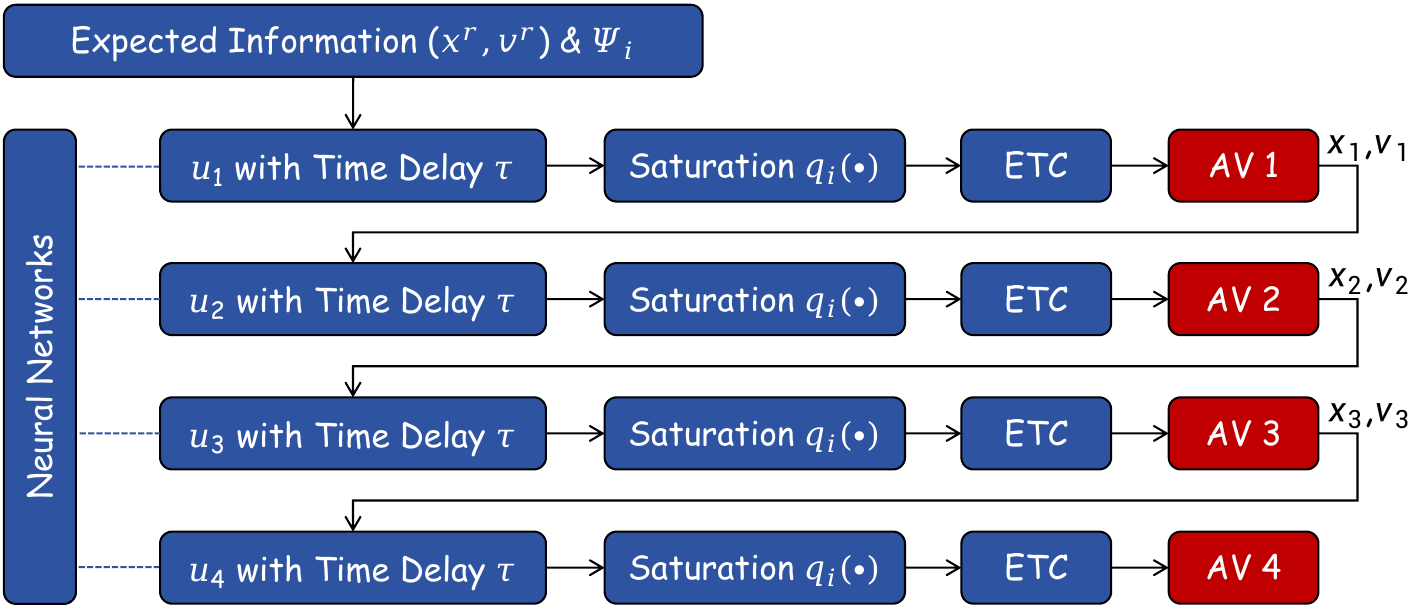}
                     \caption{The control architecture block diagram.}
                     \label{cabd}
                 \end{figure}


	\section{Problem Formulation}\label{sec2}

    \subsection{Notations}

    Let $\mathbb{R}$ and $\mathbb{R}^{n \times m}$ denote the sets of real numbers and $n \times m$ real matrices, respectively. 
    For a generic variable $\chi$, we have that $\chi^x$ and $\chi^y$ represent its longitudinal and lateral components, respectively. 
    Let $\hat{\chi}$ denote the estimate of $\chi^*$, with the estimation error being $\tilde{\chi} = \hat{\chi} - \chi^*$. 
    For a positive definite matrix $\chi$, we have that $\Pi_{min}(\chi)$ and $\Pi_{max}(\chi)$ represent its minimum and maximum eigenvalues, respectively.

	\subsection{System Model with Input Saturation and Time-delay}

        In this paper, we consider the AVs formation as a multi-agent system that consists of one virtual leader and $n$ followers. 
        Let us define the index of each follower vehicle as AV-$i$, $i=1, \cdots ,n$. 
        Then the second-order dynamical system model of AVs is designed as 
        \begin{equation}
		\label{equ1}
		\begin{aligned}
			\dot{x}_i&=v_i, \\
			\dot{v}_i&=r_i+p_i+q_i(u_i(t-\tau)), 
		\end{aligned}
	    \end{equation}
        where $x_i=[x_i^x,x_i^y]^T \in \mathbb{R}^{2\times 1}$ denotes the position of AV-$i$, $v_i=[v_i^x,v_i^y]^T\in \mathbb{R}^{2\times 1}$ denotes the velocity of AV-$i$, $m_i \in \mathbb{R}$ denotes the mass of AV-$i$, $r_i=[r_i^x/m_i,r_i^y/m_i]^T\in \mathbb{R}^{2\times 1}$ denotes the resistance and $p_i=[p_i^x/m_i,p_i^y/m_i]^T\in \mathbb{R}^{2\times 1}$ denotes the unknown external perturbation. 
        The control input with the longitudinal and lateral traction force of the vehicle engine $u_i^x$, $u_i^y$ is denoted as $u_i=[u_i^x/m_i,u_i^y/m_i]^T\in \mathbb{R}^{2\times 1}$. 
        The saturated input with a time delay $\tau$ is denoted as $q_i(\cdot)$,  and the output associated with saturation nonlinearity is represented as
        \begin{equation}
		\label{equ2}
        \begin{aligned}
		q_i(\cdot)=\text{sat}({{u}_{i}})=\left\{ \begin{matrix}
			\text{sign}({{u}_{i}}){q^*}, & | {{u}_{i}} |\geq{{q}^*},  \\
			{{u}_{i}}, & |{{u}_{i}}|<{q^*},  \\
		\end{matrix}\right.
	    \end{aligned}
        \end{equation}
        where $q^*$ represents the bound of the control input $u_i$. Notably, the scalar operations ($|u_i|$, $\text{sign}(u_i)$, $\tanh(\cdot)$) are applied to the vector $u_i$ in an element-wise manner, and this convention is adopted throughout the remainder of this subsection.
        It can be found that the sharp corners in the representation create challenges for designing an adaptive controller using the backstepping method \cite{1}. Based on \cite{1.1}, it is possible to approximate the saturation function as
        \begin{equation}
		\label{equ3}
        \begin{aligned}
        q_i(\cdot)&=\text{sat}({{u}_{i}})=\kappa_i(u_i)+\delta_i(u_i),\\
		\kappa_i(u_i)&=q^*\tanh(\frac{u_i}{q^*})=q^*\frac{e^\frac{u_i}{q^*}-e^\frac{-u_i}{q^*}}{e^\frac{u_i}{q^*}+e^\frac{-u_i}{q^*}}, 
	    \end{aligned}
        \end{equation}
        where $\delta_i$ represents the approximation error of the saturation function. 
        Define a positive constant $D_i=q^*(1-\tanh(1))$ satisfying $|\delta_i|=\text{sat}({{u}_{i}})-\kappa_i(u_i)\leq D_i$. 
        Based on the mean-value theorem~\cite{1.1}, $\kappa_i(u_i)$ is defined as $\kappa_i(u_i)=\kappa_i(u_i^*)+\partial\kappa(\cdot)/\partial u_i|_{u_i=u_i^\mu}(u_i-u_i^*)$, where $u_i^\mu=\mu u_i+(1-\mu)u_i^*$ with $0<\mu<1$. Let $u_i^*=0$, we have $\kappa_i(u_i)=\kappa^* u_i$ and $\kappa^*=\partial\kappa(\cdot)/\partial u_i|_{u_i=u_i^\mu}$. 
        According to the properties of the $tanh$ function, there exist two positive constants $\kappa_1$ and $\kappa_2$ satisfying $0<\kappa_1\leq\kappa^*\leq\kappa_2$. 
        Then system \eqref{equ1} can be represented as
        \begin{equation} 
		\label{equ4} 
		\begin{aligned} 
			\dot{x}_i&=v_i, \\
			\dot{v}_i&=r_i+p_i+\kappa^* u_i(t-\tau)+\delta_i(u_i(t-\tau)). 
		\end{aligned}
	    \end{equation}

        \subsection{Event-Triggered Strategy}


        According to \cite{2}, we introduce the dynamic-threshold event-triggered strategy to reduce the update frequency of the AVs' controllers. Define the updated controller in event-triggered strategy as $w_i(t)$ and the controller measurement error as $e_i(t)=w_i(t)-u_i(t)$. 
        The trigger condition is defined as
        \begin{equation}
		\label{equ5}
		\begin{aligned}
			w_i(t)=-(1+\rho)(U_i
			\begin{bmatrix}
				\tanh(\frac{U_iz_{i,2,1}}{\varepsilon_1})\\
				\tanh(\frac{U_iz_{i,2,2}}{\varepsilon_2})
			\end{bmatrix}
			+\overline{\rho}
			\begin{bmatrix}
				\tanh(\frac{\overline{\rho}z_{i,2,1}}{\varepsilon_1})\\
				\tanh(\frac{\overline{\rho}z_{i,2,2}}{\varepsilon_2})
			\end{bmatrix}), 
		\end{aligned}
	\end{equation}
        \begin{equation}
		\label{equ6}
		\begin{aligned}
			u_i(t)&=w_i(t_i^k),\\
			t_i^{k+1}&=\inf\{t>t_i^k \ | \ \lvert\lvert e_i(t)\rvert\rvert\geq\rho\lvert\lvert u_i(t)\lvert\lvert+\phi\}, 
		\end{aligned}
	\end{equation}
    where $\varepsilon_1$, $\varepsilon_2$, $\phi$, $0<\rho<1$ and $\overline{\rho}>\phi/(1-\rho)$ are all designed positive constants. The designed continuous controller and the total tracking error of vehicular velocity are $U_i$ and $z_{i,2}$, respectively, and will defined later in Section~\ref{sec3}. Additionally, $z_{i,2,j}$ is the $j$th component of vector $z_{i,2}$ where $j=1,2$.
    
        \subsection{Lemmas and Operational Assumptions} 
        
        To achieve the control objective, we introduce the following lemmas and assumptions that are necessary for the subsequent stability analysis.  

        \begin{lemma}[\cite{lemma2}]\label{lemma1}
        For any constant $\chi_1 \in \mathbb{R}$ and any positive constant $\chi_2 > 0$, the inequality $0\leq |\chi_1|-\chi_1\tanh\left(\frac{\chi_1}{\chi_2}\right)\leq\varepsilon^*\chi_2$ holds, where $\varepsilon^*$ is a constant satisfying $\varepsilon^*=e^{-(\varepsilon^*+1)}$ (approximately $\varepsilon^* \approx 0.2785$).
        \end{lemma}
        
        \begin{lemma}[\cite{lemma1}]\label{lemma2} 
        For any positive constant $\chi_0$ and any variable $\chi$, if $|\chi|<\chi_0$, the inequality $\log\left(\frac{\chi_0^2}{\chi_0^2-\chi}\right)\leq\frac{\chi_0^2}{\chi_0^2-\chi}$ holds.
        \end{lemma}
        
        \begin{lemma}[\cite{lemma3}]\label{lemma3}
        For any Hurwitz matrix \( E \in \mathbb{R}^{4n \times 4n} \) and any symmetric positive definite matrix \( H \in \mathbb{R}^{4n \times 4n} \), there exists a unique symmetric positive definite matrix \( F \in \mathbb{R}^{4n \times 4n} \) satisfying $E^T F + F E = -H$.
         \end{lemma}
     
       
       \begin{ass}[\cite{lemma3}]\label{ass1}
        For each AV-$i$, the position trajectory $x_i(t)$ is differentiable and Lipschitz continuous.
        Under the ETC strategy, the vehicular controller only updates its value at discrete triggering instants $t^k$. 
        The discrepancy between the sampled value $\overline{x}_i$ and the actual value $x_i$ is bounded such that for $t \in [t^k, t^{k+1})$ it holds
       \begin{equation}
       \begin{aligned}
       \label{equ7}
        &\|x_i(t^k)-x_i(t)\| \le \beta_i(t-t^k) \le \beta_{\max}(t-t^k),\\
        &\|\bar{x}_i(t^k)-x_i(t^k)\| \le \gamma_i \le \gamma_{\max},
        \end{aligned}
        \end{equation}
        where $\beta_i, \gamma_i > 0$ are positive constants, with $\beta_{\max}=\max_i\{\beta_i\}$ and $\gamma_{\max}=\max_i\{\gamma_i\}$. 
        \end{ass}
        
    
    \begin{ass}[\cite{ass1,ass2}]\label{ass2}
        The unknown external time-varying perturbation is bounded. 
        Specifically, there exists an unknown positive constant $\bar{p}_i$ such that $\|p_i(t)\| \le \bar{p}_i$ for all $t \ge 0$.
	\end{ass}

    
    Lemma~\ref{lemma1} is utilized to analysis the controller signal changes in the ETC strategy. 
    Lemma~\ref{lemma2} is applied in the formulation of the barrier Lyapunov function $V_i$. 
    Lemma~\ref{lemma3} is required for the key scaling steps when differentiating the Lyapunov function. 
    Finally, Assumptions~\ref{ass1} and~\ref{ass2} are employed to guarantee that the system perturbations and observation errors are strictly bounded. 
    
    \subsection{Control Architecture}\label{control_arch_subsec}
    The fleet consists of $n$ follower AVs and one virtual leader. The virtual leader broadcasts its desired trajectory $x^r(t)$ to all followers. Each follower AV-$i$ independently computes its control input using only its own state $(x_i, v_i)$, the received leader information, the prescribed formation offset $\Psi_i$ and local algorithms (neural networks, uncertainty observer, and event-triggered strategy). No direct communication occurs between followers, the control architecture is therefore distributed.
    
     \subsection{Problem Statement}
     Consider a multi-agent system consisting of $n$ follower AVs and one virtual leader. 
        The dynamics of each AV are governed by the second-order system \eqref{equ1}, which accounts for both longitudinal and lateral motion. 
        During operation, each vehicle is subject to unknown mass-normalized resistance $r_i$, time-varying external perturbations $p_i$, input saturation $q_i(\cdot)$, and communication time delays $\tau$. 
        Let $x^r(t) \in \mathbb{R}^{2\times1}$ and $v^r(t) \in \mathbb{R}^{2\times1}$ denote the desired trajectory and velocity of the virtual leader, respectively. 
        The desired position for the $i$-th follower to maintain the formation is defined as $x_i^r(t) = x^r(t) + \Psi_i$, where $\Psi_i \in \mathbb{R}^{2\times1}$ represents the prescribed formation offset for AV-$i$. 
        We define the position tracking error as $\mathcal{Z}_{i,1}(t) = x_i(t) - x_i^r(t)$.
        The control updates for each vehicle are governed by the dynamic-threshold event-triggered strategy defined in \eqref{equ5} and \eqref{equ6}. 
        The control architecture diagram is shown in Fig.~\ref{cabd}.
        Under these conditions, the control objectives of this paper are formulated as follows:
        \begin{enumerate}
            \item Design a distributed adaptive control law $u_i(t)$ such that all closed-loop signals remain bounded, and the position tracking error $\mathcal{Z}_{i,1}(t)$ is uniformly ultimately bounded (UUB). That is, $\lim_{t \to \infty} \|\mathcal{Z}_{i,1}(t)\| \le \iota$ for a small positive constant $\iota$, despite the presence of input saturation, unknown perturbations, and time delays.
            \item Ensure that the event-triggered strategy \eqref{equ6} significantly reduces the frequency of control updates from the continuous controller to the actuators, thereby conserving communication and computational resources while strictly avoiding Zeno behavior (i.e., ruling out an infinite number of triggers in a finite time interval).
        \end{enumerate}

	\section{Controller Design}\label{sec3}
    Following the control architecture described in Section~\ref{control_arch_subsec}, this section details the controller design for each follower AV-$i$. 
    The design proceeds via the backstepping method, incorporating Radial Basis Function Neural Networks (RBF NNs) to approximate unknown resistance, a command filter to avoid complexity explosion, an auxiliary system to compensate for input delay and saturation, and an uncertainty observer to handle external disturbances.
    
    To handle the unknown resistance $r_i$, we introduce RBF NNs. Note that, according to the universal approximation theorem~\cite{lemma3,ass1,ass2}, an unknown continuous function $f_i(\chi)$ can be approximated over a compact set $\Omega_{\chi}$ as $f_i(\chi)=W_i^{*T}K_i(\chi)+\varrho_i$, where, $W_i^*=[W_{i,1}^*, W_{i,2}^*,\dots, W_{i,n}^* ]^T \in \mathbb{R}^n$ denotes the optimal weight vector that minimizes the approximation error $\varrho_i$, and $K_i(\chi)=[k_{i,1}(\chi), k_{i,2}(\chi),\dots, k_{i,n}(\chi)]^T \in \mathbb{R}^n$ is the basis function vector. The $j$-th basis function is chosen as a Gaussian function $k_{i,j}(\chi)=\exp(-\|\chi-\epsilon_{i,j}\|^2/\zeta^2)$, where $\epsilon_{i,j}$ is the center of the receptive field and $\zeta$ is the width of the Gaussian function. 
    
 
 We utilize the RBF NNs to approximate the resistance term $r_i$ in \eqref{equ1} such that $g_i r_i = W_i^{*T}K_i(v_i)+\varrho_i$, where $g_i$ is a positive design constant. 
 Substituting this into \eqref{equ1} yields
    \begin{equation}
		\label{equ8}
		\begin{aligned}
			\dot{x}_i&=v_i, \\
			\dot{v}_i&=\frac{1}{g_i}W_i^{*T}K_i(v_i)+d_i+q_i(u_i(t-\tau)),
		\end{aligned} 
	\end{equation} 
    where $d_i = \frac{1}{g_i}\varrho_i + p_i$ denotes the lumped external uncertainty. 
    To handle this uncertainty, we introduce an uncertainty observer with an auxiliary variable $h_i=d_i-g_iv_i$. 
    The estimate of the uncertainty $d_i$ is expressed as $\hat{d}_i=\hat{h}_i+g_iv_i$. 
    We also define the uncertainty estimation error as $\tilde{h}_i=h_i-\hat{h}_i$ and establish a corresponding positive definite Lyapunov function candidate $V_{i,0}=\frac{1}{2}\tilde{h}_i^2$. 


To overcome the explosion of complexity problem inherent in traditional backstepping, caused by the repeated analytical differentiation of the virtual controller $\alpha_i$, we introduce a second-order command filter:
    \begin{equation}
		\label{equ9}
		\begin{aligned}
			\dot{x}^*_i&=c_i v^*_i, \\
			\dot{v}^*_i&=-2\Delta_i c_i v^*_i - c_i(x^*_i-\alpha_i),
		\end{aligned}
	\end{equation}
    where $c_i>0$ and $\Delta_i \in (0,1)$ are design constants. The variables $x^*_i$ and $v^*_i$ represent the internal states of the command filter, and $\alpha_i$ is the virtual control input. The initial states of the filter are set to $x^*_i(0)=\alpha_i(0)$ and $v^*_i(0)=0$. Furthermore, to mitigate the adverse effects of the input delay and saturation, we construct the following auxiliary system:
    \begin{equation}
		\label{equ10}
		\begin{aligned}
			\dot{\psi}_{i,1}&={\psi}_{i,2}-b_{i,1}{\psi}_{i,1}, \\
		    \dot{\psi}_{i,2}&=-b_{i,2}{\psi}_{i,2}-q_i(u_i)+q_i(u_i(t-\tau)),
		\end{aligned}
	\end{equation}
    where $b_{i,1}>\frac{1}{2}$ and $b_{i,2}>\frac{3}{2}$ are positive design constants, and the initial conditions are ${\psi}_{i,1}(0)={\psi}_{i,2}(0)=0$.

 Next, we define the compensated tracking errors $z_{i,1}$ and $z_{i,2}$, which incorporate the original tracking errors $\xi_{i,1}$ and $\xi_{i,2}$, and the command filter compensation signals $\eta_{i,1}$ and $\eta_{i,2}$:
    \begin{equation}
		\label{equ11}
		\begin{aligned}
			{z}_{i,1}&=\xi_{i,1}-\eta_{i,1}, \\
			{z}_{i,2}&=\xi_{i,2}-\eta_{i,2}.
		\end{aligned}
	\end{equation}
    
    The original tracking errors are defined as $\xi_{i,1}=x_i-x_i^r-\psi_{i,1}$ and $\xi_{i,2}=v_i-v^*_{i}-v_i^r-\psi_{i,2}$, where $x^r_i$ and $v^r_i$ are the desired reference position and velocity, respectively. 
    To compensate for the filtering errors, the compensation signals are generated by $\dot{\eta}_{i,1}=-o_{i,1}\eta_{i,1}+\eta_{i,2}+(v^*_i-\alpha_i)$ and $\dot{\eta}_{i,2}=-o_{i,2}\eta_{i,2}-\eta_{i,1}$, where $o_{i,1}, o_{i,2} > 0$ are design constants.
    
    Finally, the virtual controller $\alpha_i$, the continuous-time nominal controller $U_i$, and the adaptive update law for the NNs weights $\hat{W}_i^*$ are designed as follows:
    \begin{equation}
		\label{equ12}
		\begin{aligned}
			\alpha_i=&-o_{i,1}\xi_{i,1}-b_{i,1}\psi_{i,1}+\hat{v}_i^r, \\
			U_i=&-o_{i,2}\xi_{i,2}-b_{i,2}\psi_{i,2}+\hat{v}_i^r-\eta_{i,1}-\hat{d}_i\\
            &-\frac{z_{i,2}}{\sigma_{i,2}^2-z_{i,2}^2}-\frac{1}{g_i}\hat{W}_i^{*T}K_i(v_i), \\
            \dot{\hat{W}}_i^*=&s_i\frac{z_{i,2}\hat{W}_i^{*T}}{\sigma_{i,2}^2-z_{i,2}^2}K_i(v_i)-s_i l_i \hat{W}_i^*,
		\end{aligned}
	\end{equation}
    where $\sigma_{i,2}, s_i, l_i$ are positive design constants. 
    In \eqref{equ12}, $o_{i,1}$ and $o_{i,2}$ act as the primary controller gains $b_{i,1}$ and $b_{i,2}$ are the time-delay compensation gains, $s_i$ and $l_i$ denote the adaptation gain and leakage factor for the NNs, respectively, $\hat{v}_i^r$ represents the feedforward reference velocity, and $\hat{d}_i$ is the uncertainty compensation term provided by the observer.
    
	\section{Stability Analysis}\label{sec4}


    In this section, we prove the stability of the closed-loop second-order dynamical system utilizing the proposed adaptive event-triggered formation control framework, while formally verifying the absence of Zeno behavior. 


    \begin{theorem}\label{main_theorem_conv}
        Consider the AV second-order dynamics \eqref{equ1}, the continuous-time control law \eqref{equ12}, and the relative-threshold event-triggering mechanism \eqref{equ5}, \eqref{equ6}. Then, the vehicle fleet is able to track the prescribed position and velocity trajectories, while all closed-loop signals remain bounded. In addition, there exists a constant $t^{*}>0$ such that all inter-execution intervals $\{t^{k+1}-t^{k}\}$ are uniformly lower bounded by $t^{*}$.
\end{theorem}

\begin{proof}
We design the Lyapunov function as
\begin{equation}
		\label{equ13}
		\begin{aligned}
            V=&\sum\limits_{i=1}^{n} V_i, \\
			V_i=&\frac{1}{2}\log\frac{\sigma^2_{i,1}}{\sigma^2_{i,1}-z_{i,1}^2}+\frac{1}{2}\log\frac{\sigma^2_{i,2}}{\sigma^2_{i,2}-z_{i,2}^2}\\&+\frac{1}{2\Upsilon_i}\tilde{W}_i^{*T}\tilde{W}_i^{*}+V_{i,0}, 
		\end{aligned}
	\end{equation}
    where $\sigma_{i,1}$, $\sigma_{i,2}$, $\Upsilon_i$ are designed positive constants. 
    Here $\sigma_{i,1}$ and $\sigma_{i,2}$ are barrier-Lyapunov-function based boundary parameters. 
    According to the event-triggered strategy \eqref{equ5}, we define two time-varying constants $|\Xi_{i,1}(t)|\leq1$ and $|\Xi_{i,2}(t)|\leq1$, satisfying $\Xi_{i,1}(t^k)=\Xi_{i,2}(t^k)=0$, $\Xi_{i,1}(t^{k+1})=\Xi_{i,2}(t^{k+1})=\pm1$, and we have $w_i(t)=(1+\Xi_{i,1}\rho)\mu_i(t)+\Xi_{i,2}\phi$. 
    Based on Lemma~\ref{lemma1}, we have
    \begin{equation}
		\label{equ14}
		\begin{aligned}
			z_{i,2}\frac{w_i(t)-\Xi_{i,2}\phi}{1+\Xi_{i,1}\delta}\leq 0.557(\varepsilon_1+\varepsilon_2). 
		\end{aligned}
	\end{equation}
    
    Furthermore, based on the Young’s inequality and Lemma~\ref{lemma2}, we have 
    \begin{equation}
		\label{equ15}
		\begin{aligned}
        l_i\tilde{W}_i^{*}\hat{W}_i^{*}&\leq-\frac{l_i}{2}||\tilde{W}_i^*||^2+\frac{l_i}{2}||W_i^*||^2,\\
        -\frac{o_{i,1}\sigma^2_{i,1}}{\sigma^2_{i,1}-z_{i,1}^2}&\leq-o_{i,1}\log\frac{\sigma^2_{i,1}}{\sigma^2_{i,1}-z_{i,1}^2},\\
        -\frac{o_{i,2}\sigma^2_{i,2}}{\sigma^2_{i,2}-z_{i,2}^2}&\leq-o_{i,2}\log\frac{\sigma^2_{i,2}}{\sigma^2_{i,2}-z_{i,2}^2},\\
			\frac{z_{i,2}\delta_i}{\sigma_{i,2}^2-z_{i,2}^2}&\leq\frac{z_{i,2}^2}{2(\sigma_{i,2}^2-z_{i,2}^2)}+\frac{\delta_i^2}{2},\\
              \frac{z_{i,2}{\hat{W}}_i^{*T}}{\sigma_{i,2}^2-z_{i,2}^2}K_i(v_i)\tilde{W}_i^{*T}&\leq\frac{1}{2}\frac{z_{i,2}^2{||\hat{W}}_i^{*T}||^2}{(\sigma_{i,2}^2-z_{i,2}^2)^2}+\frac{1}{2}||\tilde{W}_i^*||^2. \\
		\end{aligned}
	\end{equation}
    
    \begin{figure}
            \centering
            \subfloat[0s: Initial states for all AVs.]{\includegraphics[width=0.9\linewidth]{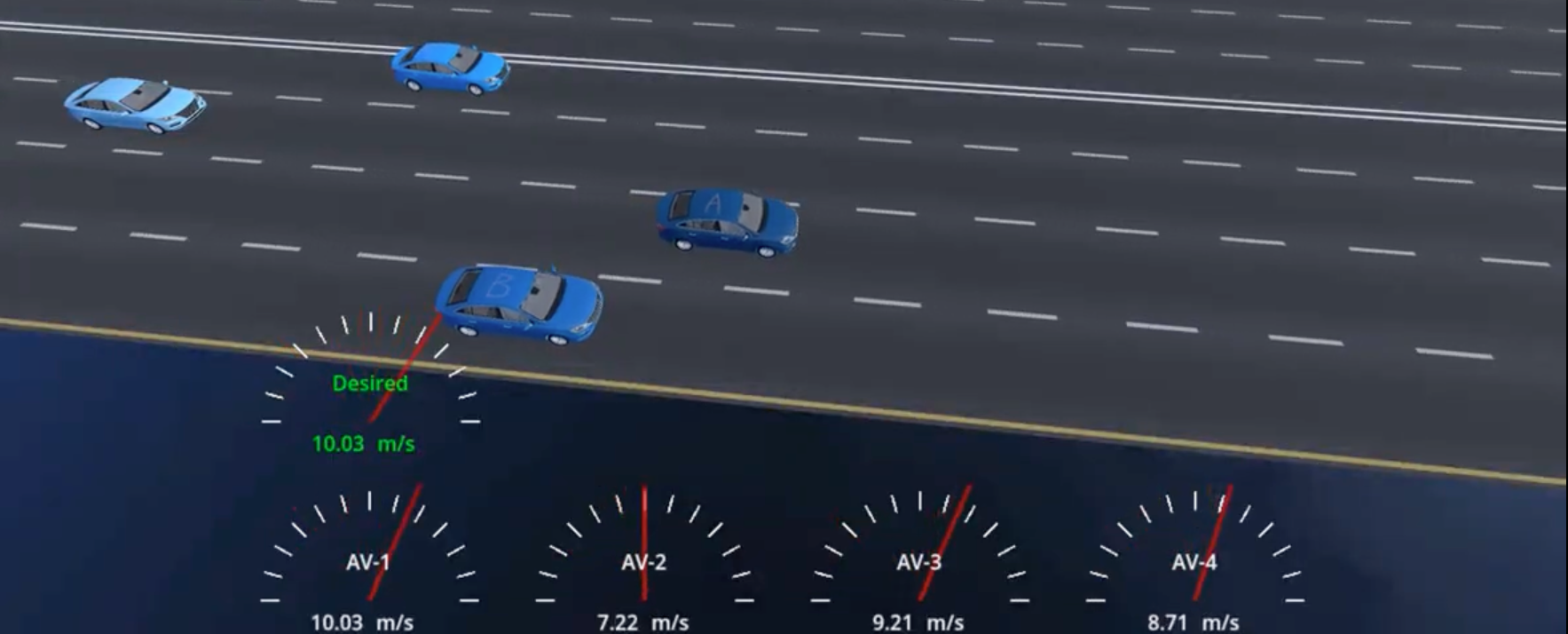}}\\
            \subfloat[0s - 20s: AVs transition to the linear formation.]{\includegraphics[width=0.9\linewidth]{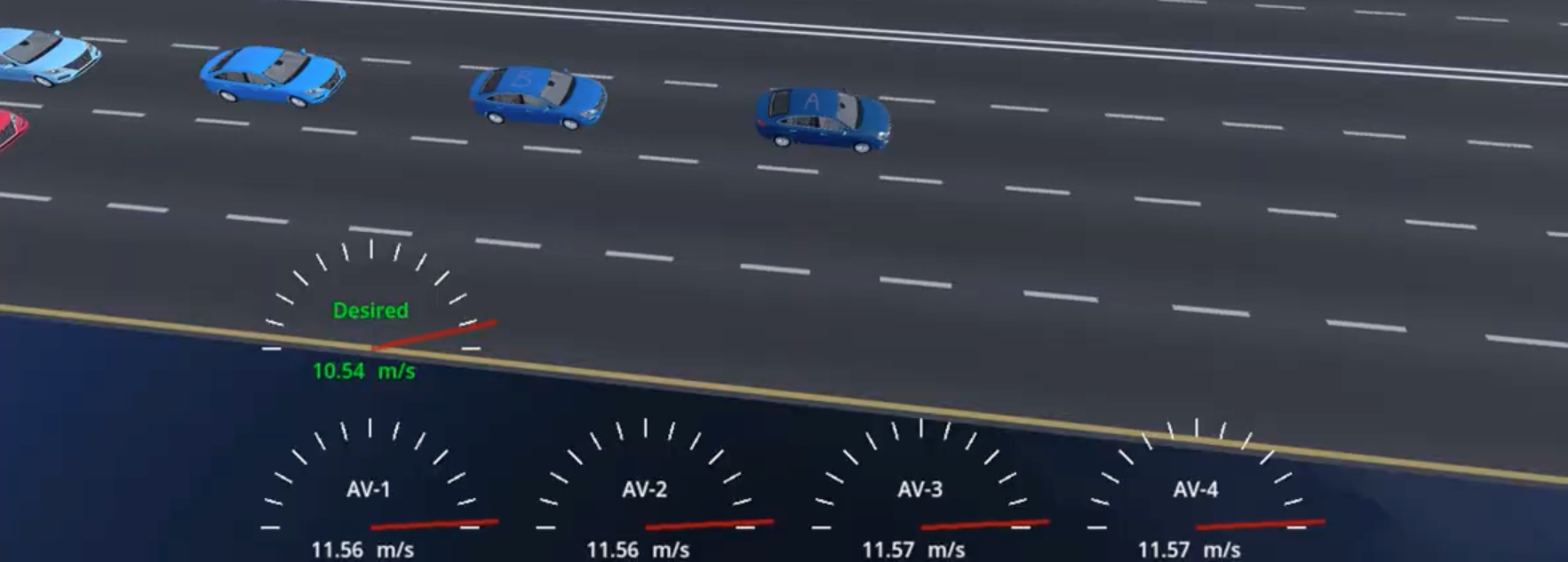}}\\
            \subfloat[20s - 22s: AVs decelerate rapidly from 12m/s to 4m/s.]{\includegraphics[width=0.9\linewidth]{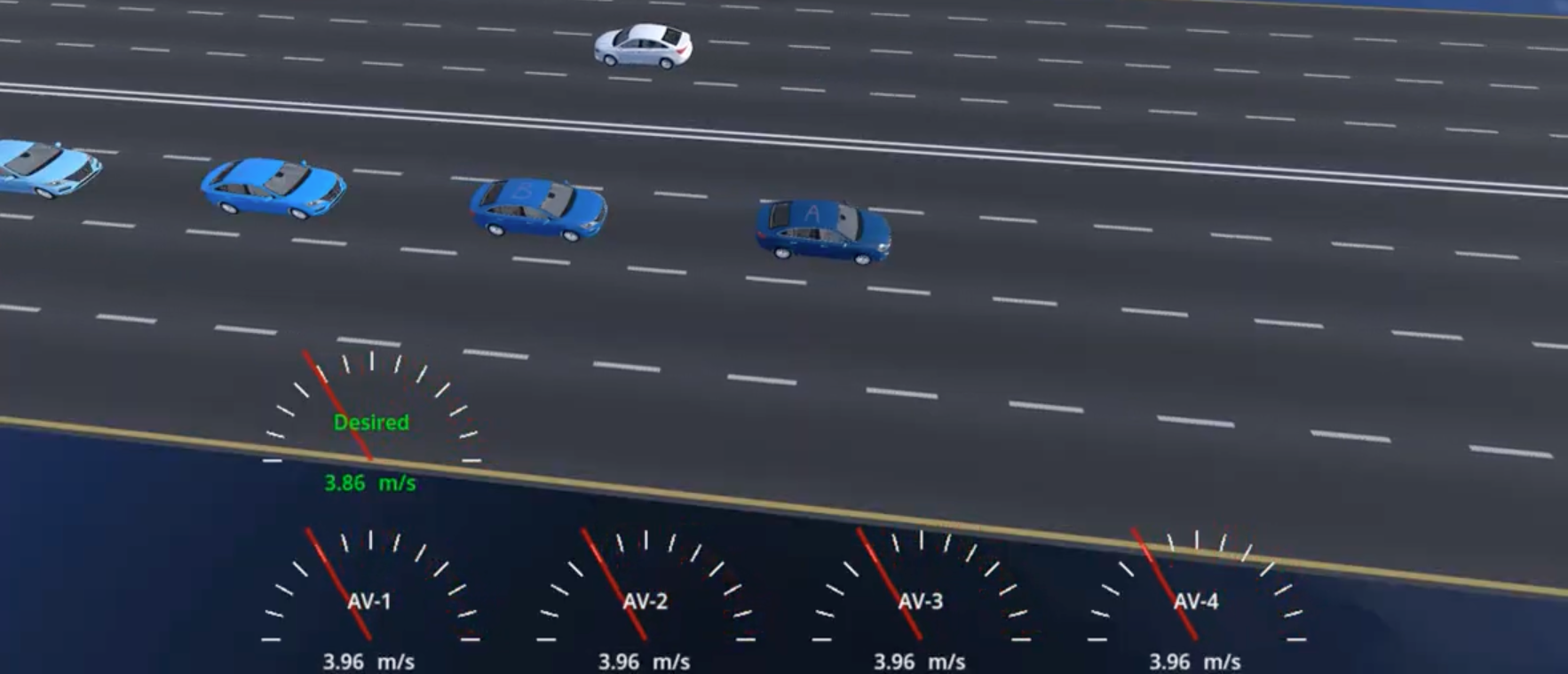}}\\
            \subfloat[22s- 30s: AVs decelerate and prepare to form a square formation.]{\includegraphics[width=0.9\linewidth]{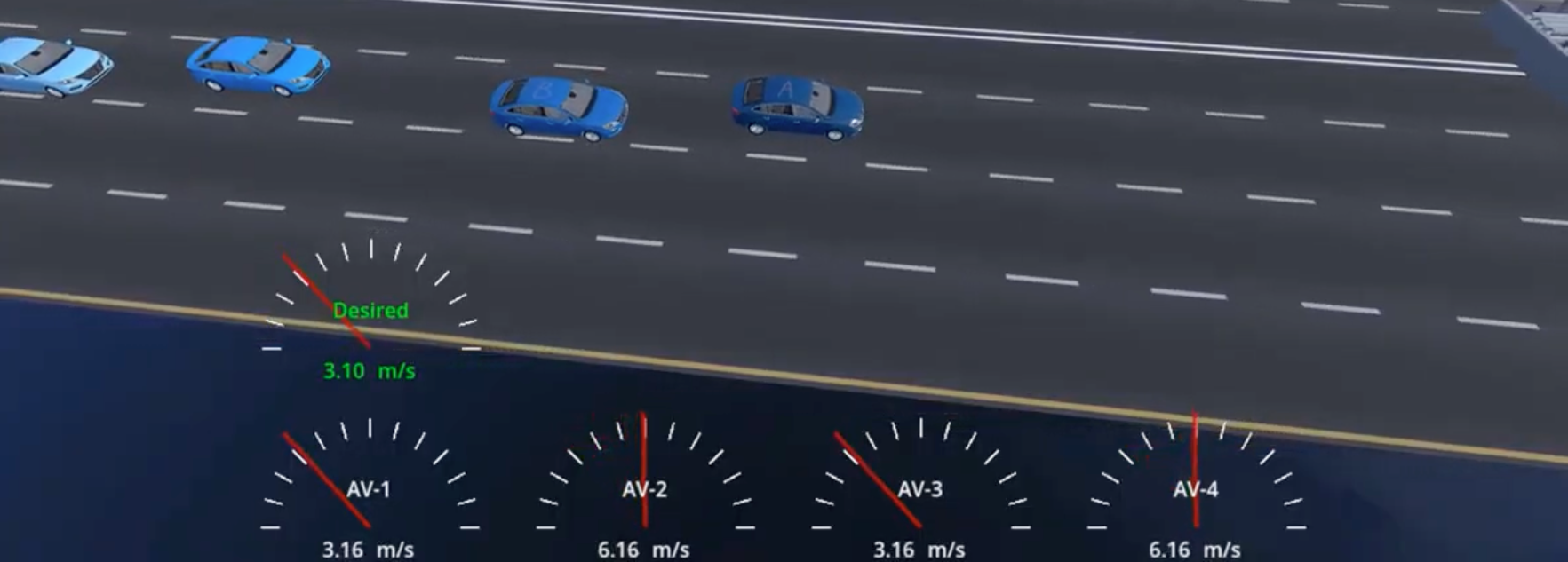}}\\
            \subfloat[30s- 40s: AVs transition to the square formation.]{\includegraphics[width=0.9\linewidth]{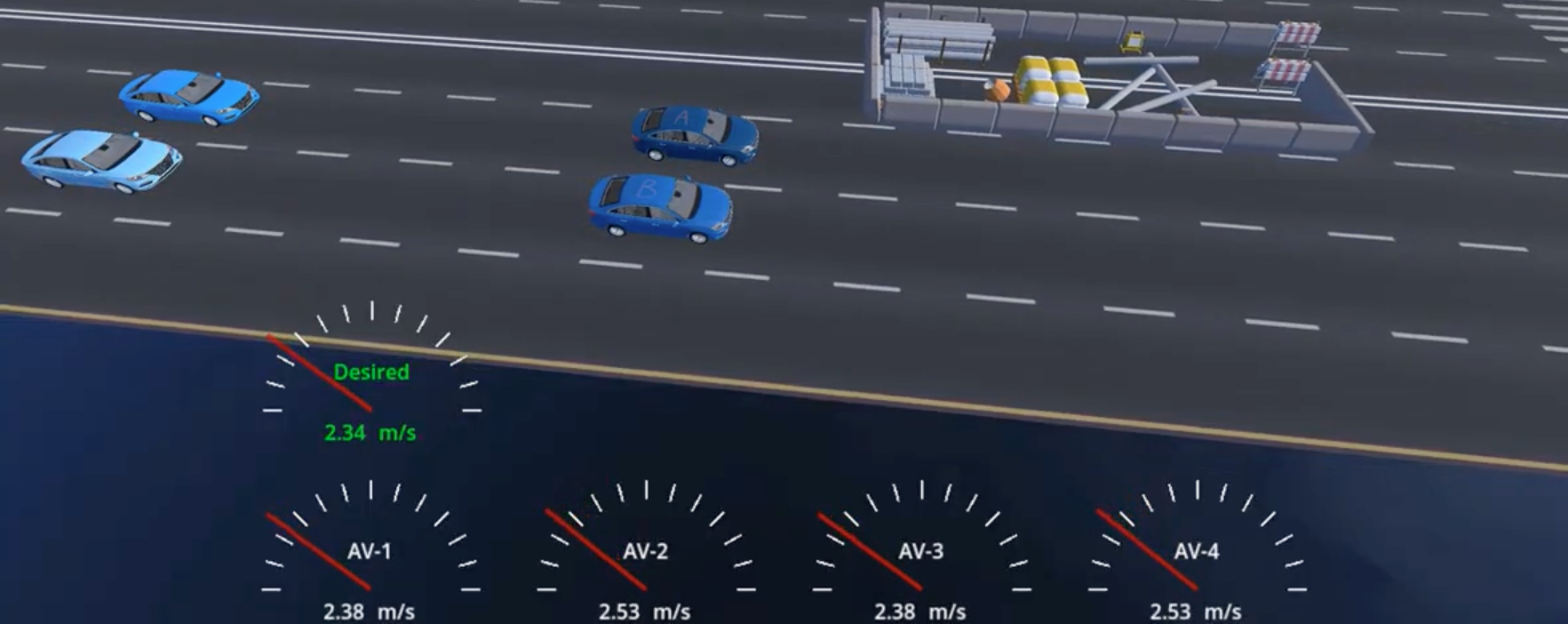}}\\
            \subfloat[40s - 50s: AVs decelerate and reach the final state.]{\includegraphics[width=0.9\linewidth]{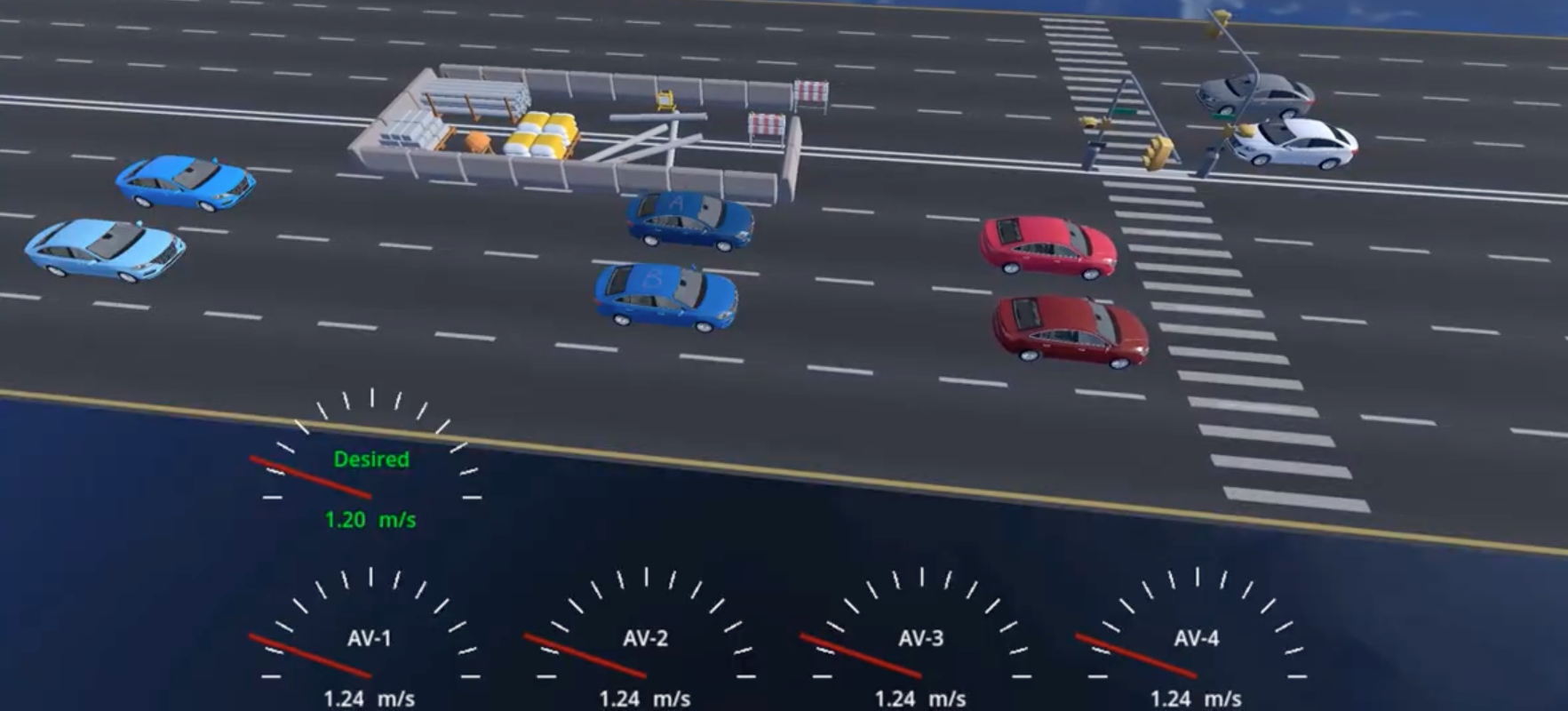}}\\
            \caption{3D visualized video shows formation change process.}
            \label{3D}
        \end{figure}

    According to \cite{stability} and Assumption~\ref{ass2}, we define $d^*$ as the upper bound of $||\dot{d}_i||$. Based on Assumption~\ref{ass1} and Lemma~\ref{lemma3}, taking the derivative of $V$ in \eqref{equ13}, we have
    \begin{equation}
		\label{equ16}
		\begin{aligned}
			\dot{V}\leq& -\sum_{i=1}^{n}o_{i,1}\log\frac{\sigma^2_{i,1}}{\sigma^2_{i,1}-z_{i,1}^2}-\sum_{i=1}^{n}o_{i,2}\log \frac{\sigma^2_{i,2}}{\sigma^2_{i,2}-z_{i,2}^2}\\
            &-\sum_{i=1}^{n}\frac{l_i}{2}||\tilde{W}_i^*||^2+\sum_{i=1}^{n}\frac{l_i}{2}||W_i^*||^2+\sum_{i=1}^{n}0.557(\varepsilon_1+\varepsilon_2)\\
            &+\sum_{i=1}^{n}\frac{D_i^2}{2}+\frac{1}{2\pi_{i}}||\tilde{W}_i^*||^2+||x_i^2||\frac{g_i^2+g_i^4}{2}+\frac{d^{*2}}{2}-\varpi_i\tilde{h}_i^2, 
		\end{aligned}
	\end{equation}
    where $\pi_{i}$ and $\varpi_i$ are the designed positive constants with $\varpi_i=g_i-\frac{3}{2}-\frac{\pi_{i}^2}{2}$. 
    We now define $\Pi_{min}(\chi)$ that denotes the minimum eigenvalue of the positive definite matrix $\chi$. Then, we have
    \begin{equation}
		\label{equ17}
		\begin{aligned}
			\dot{V}\leq& -\varphi V+\lambda, 
		\end{aligned}
	\end{equation}
    where $\varphi=\min\{\Pi_{min}(2o_{i,1}), \Pi_{min}(2o_{i,2}), \Pi_{min}(l_i)\}$ and $\lambda=\sum_{i=1}^{n}\frac{D_i^2}{2}+\sum_{i=1}^{n}0.557(\varepsilon_1+\varepsilon_2)+\frac{1}{2\pi_{i}}||\tilde{W}_i^*||^2+||x_{i,1}^2||\frac{c_i^2+c_i^4}{2}+\frac{\xi^{*2}}{2}$. Then \eqref{equ13} satisfies
    \begin{equation}
		\label{equ18}
		\begin{aligned}
			0\leq V(t)\leq\frac{\lambda}{\varphi}+(V(0)-\frac{\lambda}{\varphi})e^{-\varphi t}. 
		\end{aligned}
	\end{equation}
    
    From \eqref{equ18}, the total tracking errors $z_{i,1}$, $z_{i,2}$ and parameter estimation errors $\tilde{W}_{i}$ are all UUB. 
    Since all close-loop error signals are bounded and based on \eqref{equ5}, we have that $w_i(t)$ is continuous and bounded. 
    Furthermore, observing that $e_i(t^k)=0$ and $\lim_{t\rightarrow t^{k+1}}e_i(t)=\rho|u_i|+\phi$, there exists a constant $\Theta_i \in \mathbb{R}^+$ ensuring that the event-triggered minimum inter-execution interval $t^*$ satisfies $t^*\geq(\delta|u_i|+\phi)/\Theta_i$, successfully preventing Zeno behavior.
   

    In Theorem~\ref{main_theorem_conv}, we established that the proposed collaborative adaptive formation control framework enables the AV fleet to achieve consensus tracking while keeping all control errors bounded. 
    Specifically, the proof introduces the Lyapunov function in \eqref{equ13}, which is constructed based on the tracking errors. 
    Then, by applying Lyapunov stability analysis, we showed through \eqref{equ17} and \eqref{equ18} that all errors are UUB. 
    Compared with the previous work in \cite{lemma3}, the present analysis differs in \eqref{equ14} and \eqref{equ15} because the control law in \eqref{equ12} is specifically designed to handle input saturation and time delays. 
    In addition, unlike the static triggering threshold adopted in~\cite{lemma3}, the threshold in our proposed event-triggering condition varies with $u_i(t)$. 
    Moreover, the adopted ETC strategy is shown to exclude Zeno behavior.
\end{proof}

    \begin{figure*}
            \centering
            \includegraphics[width=0.49\textwidth]{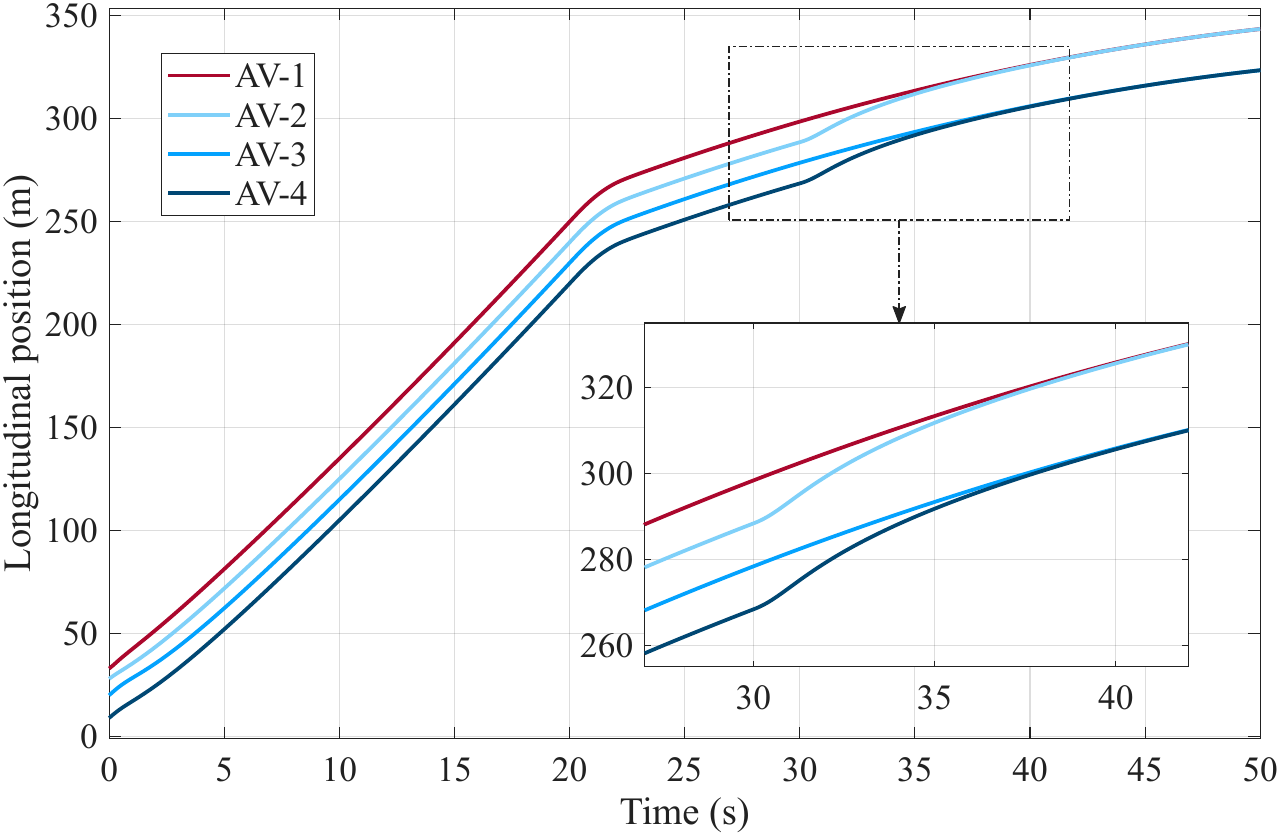}
            \includegraphics[width=0.49\textwidth]{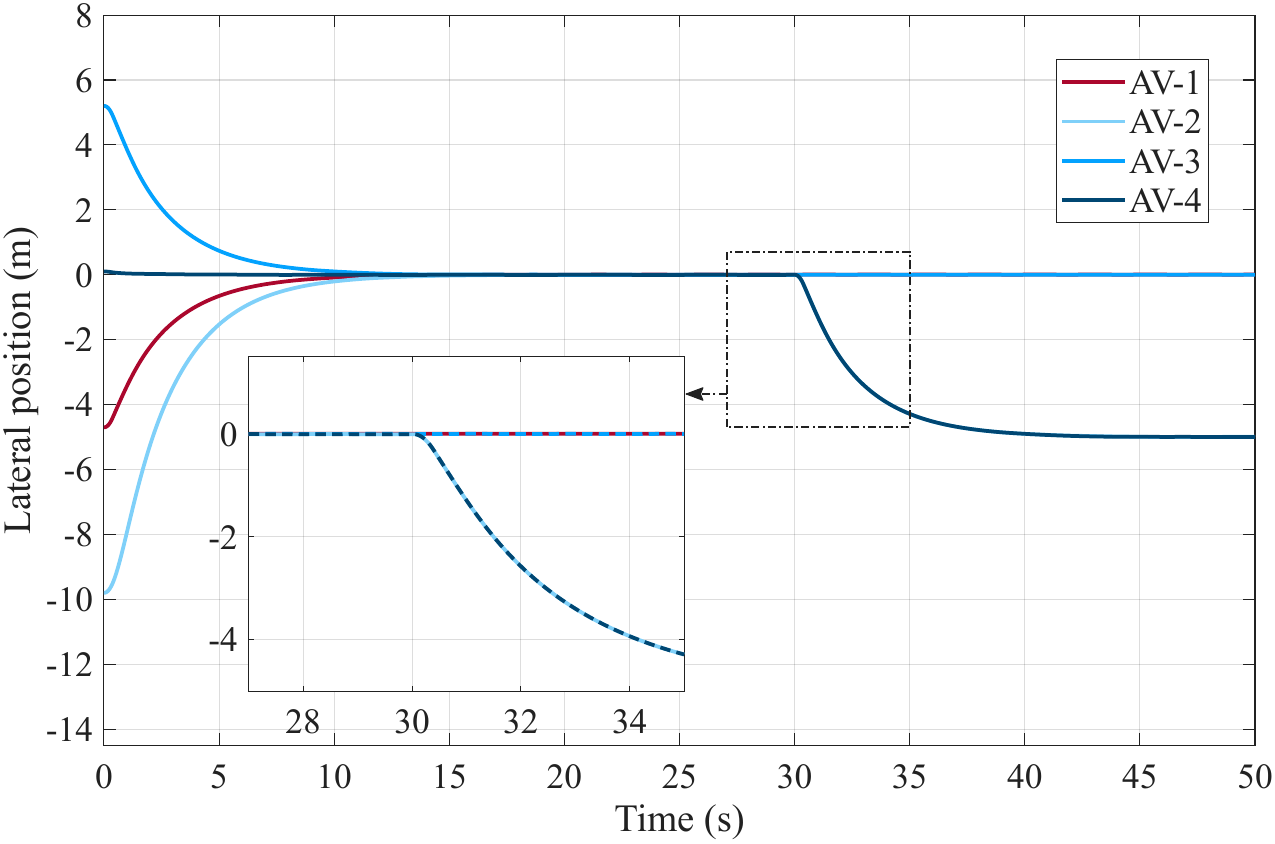}
            \caption{The longitudinal and lateral tracking control performance of the AV-$i$.}
            \label{tracking}
    \end{figure*}
    
    \begin{figure*}
            \centering
            \includegraphics[width=0.49\textwidth]{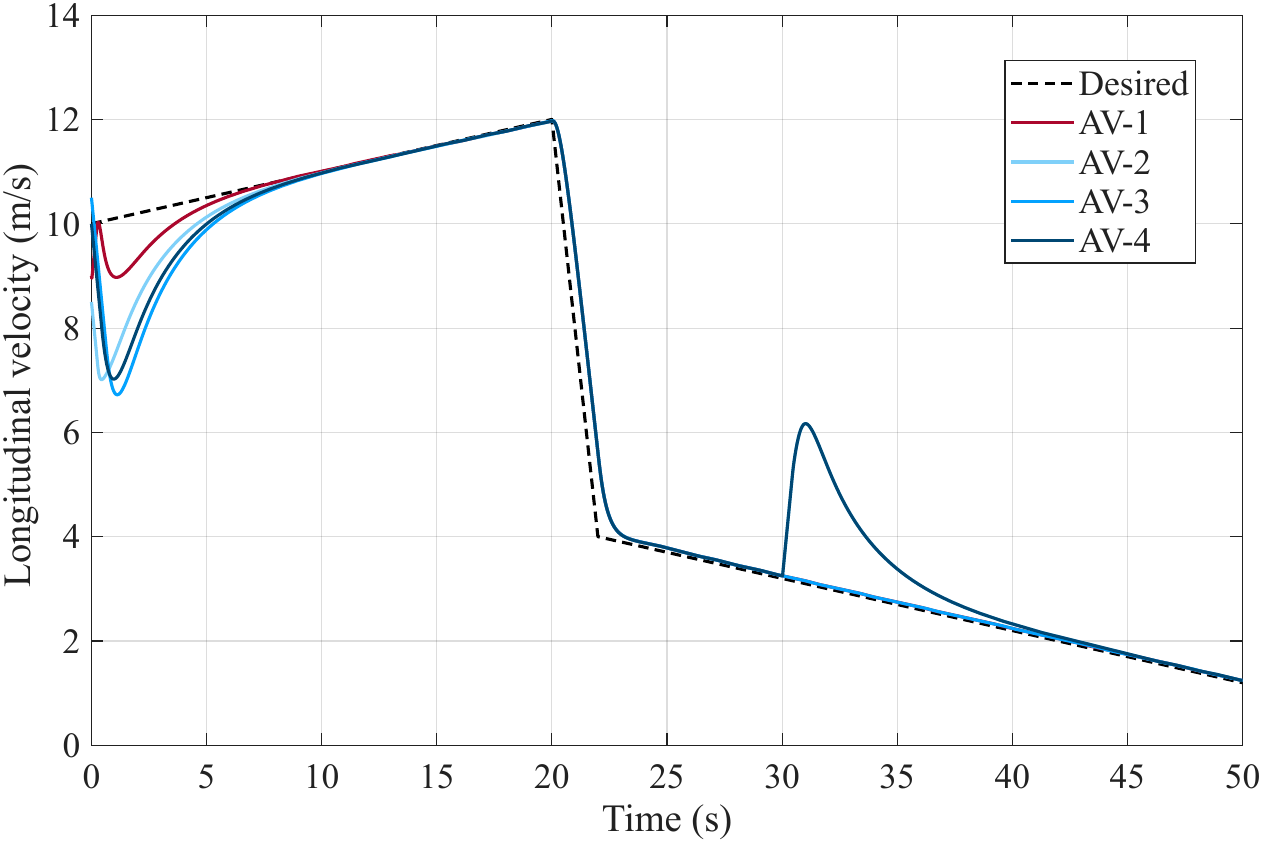}
            \includegraphics[width=0.49\textwidth]{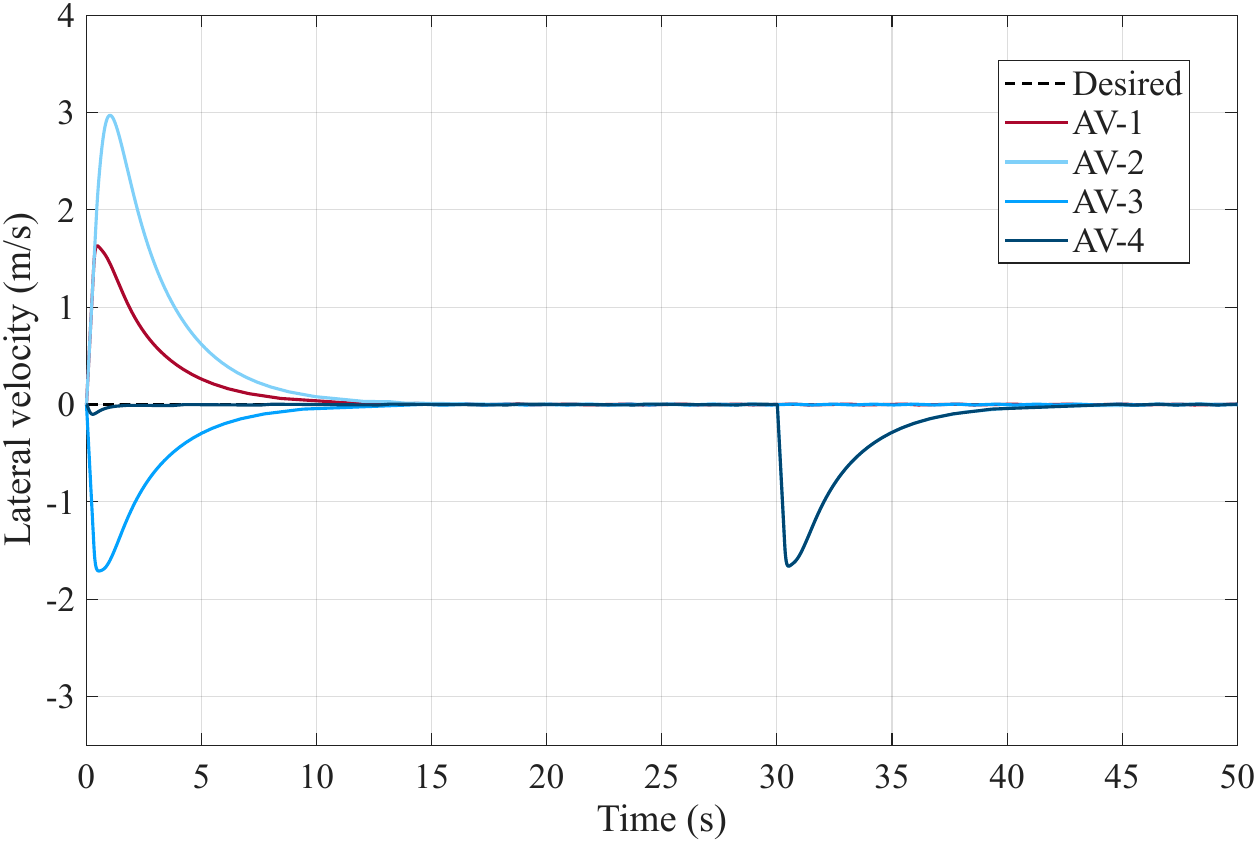}
            \caption{The longitudinal and lateral velocities of AV-$i$ under relative-threshold event-triggered control strategy.}
            \label{velocity}
    \end{figure*}

    \section{Simulation Results}\label{sec5}
	
     To verify the validity of the theoretical derivation, we conduct numerical simulations in MATLAB to verify the effectiveness of the proposed control strategy, focusing on a platoon formation change scenario involving four AVs on a four-lane road. Under the proposed collaborative adaptive formation control framework, AVs start from different initial positions and track the leader under speed variations and formation reconfigurations, thereby evaluating their dynamic responses and tracking performance.
    In this section, we first introduce the parameter setting for the control strategy. Then we present the simulation results, including the position tracking errors (Fig.~\ref{tracking}), the velocity profiles (Fig.~\ref{velocity}), the saturated control inputs (Fig.~\ref{controller}), the event‑triggering statistics and update intervals (Table~\ref{triggercounts} and Fig.~\ref{trigger}), and the inter‑vehicle safe distances (Fig.~\ref{safe}). In addition, we discuss the implications for traffic flow. Finally, we present the comparisons with existing literature.

        \subsection{Parameter Setting}
        
        In the simulation, the masses of the AVs are defined as $m_1 = 2450~\mathrm{kg}$, $m_2 = 2135~\mathrm{kg}$, $m_3 = 2980~\mathrm{kg}$, $m_4 = 2370~\mathrm{kg}$. 
        According to \cite{2025}, the resistance is characterized by $r_i^{x}=0.5H_1 \times H_2 \times H_3 \times (v_i^{x})^2$ and $r_i^{y}=0.5H_1 \times H_2 \times H_3 \times (v_i^{y})^2$, with air density $H_1=1.206\ \mathrm{kg/m^3}$, cross-sectional area $H_2 = 5.58~\mathrm{m^2}$ and dimensionless drag coefficient $H_3 = 0.3$. 
        Then define the vehicular leader's desired position trajectory $x^r(t)=[x_1^{rx},x_1^{ry}]^T$ as
    \begin{align}
	x_1^{rx} &=
	\begin{cases}
		0.05t^2+10t+35\ \mathrm{m}, & 0\ \mathrm{s} \leq t < 20\ \mathrm{s}, \\
		-2t^2+92t-785\ \mathrm{m}, & 20\ \mathrm{s} \leq t < 22\ \mathrm{s}, \\
		-0.05t^2+6.2t+158.8\ \mathrm{m}, & 22\ \mathrm{s} \leq t \leq 50\ \mathrm{s}, 
	\end{cases}
	\label{equ23} \\
	x_1^{ry} &= 0\ \mathrm{m/s}, \quad 0\ \mathrm{s} \leq t < 50\ \mathrm{s} . 
	\label{equ24}
\end{align}

    The four lanes are each $5\mathrm{m}$ wide, with the centerline of the second lane from the top serving as the x-axis. The initial states of AVs' position are $x_1(0)=[35 \mathrm{m},-4.8 \mathrm{m}]^T$,
    $x_2(0)=[30 \mathrm{m},-9.9 \mathrm{m}]^T$,
    $x_3(0)=[21 \mathrm{m},5.3 \mathrm{m}]^T$,
    $x_4(0)=[10 \mathrm{m},-0.2 \mathrm{m}]^T$. 
    The initial states of AVs' speed are 
    $v_1(0)=[9 \mathrm{m/s},0]^T$,
    $v_2(0)=[8.5 \mathrm{m/s},0]^T$,
    $v_3(0)=[10.5 \mathrm{m/s},0]^T$,
    $v_4(0)=[10 \mathrm{m/s},0]^T$.

    The prescribed formation offset $\Psi_i$ is defined as:
    From $0\mathrm{s}$ to $30\mathrm{s}$, we set $\Psi_1=[-10\mathrm{m},0]^T$, $\Psi_2=[-20\mathrm{m},0]^T$, $\Psi_3=[-30\mathrm{m},0]^T$ and $\Psi_4=[-40\mathrm{m},0]^T$; From $30\mathrm{s}$ to $50\mathrm{s}$, we set $\Psi_1=[-10\mathrm{m},0]^T$, $\Psi_2=[-10\mathrm{m},-5\mathrm{m}]^T$, $\Psi_3=[-30\mathrm{m},0]^T$ and $\Psi_4=[-30\mathrm{m},-5\mathrm{m}]^T$.

    The control parameters are selected as follows:
    the input saturation boundary is $q^*=4.5$; the controller gains are $o_{i,1}=9$, $o_{i,2}=1.2$; the time-delay damping-like gains are $b_{i,1}=b_{i,2}=1.5$; the observer gain is $g_i=1.2$; the NNs adaptation gain and NNs leakage gain are $s_i=8$ and $l_i=5$; the dynamic-threshold ETC strategy constants are $\rho=0.01$, $\phi=3$, $\varepsilon_1=\varepsilon_2=4$.

    \begin{figure*}
            \centering
            \includegraphics[width=0.475\linewidth]{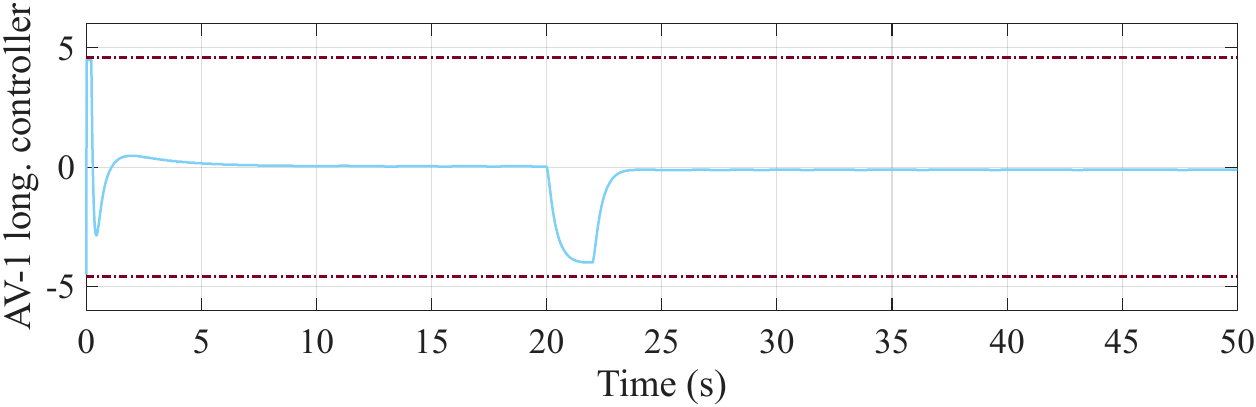}
            \includegraphics[width=0.475\linewidth]{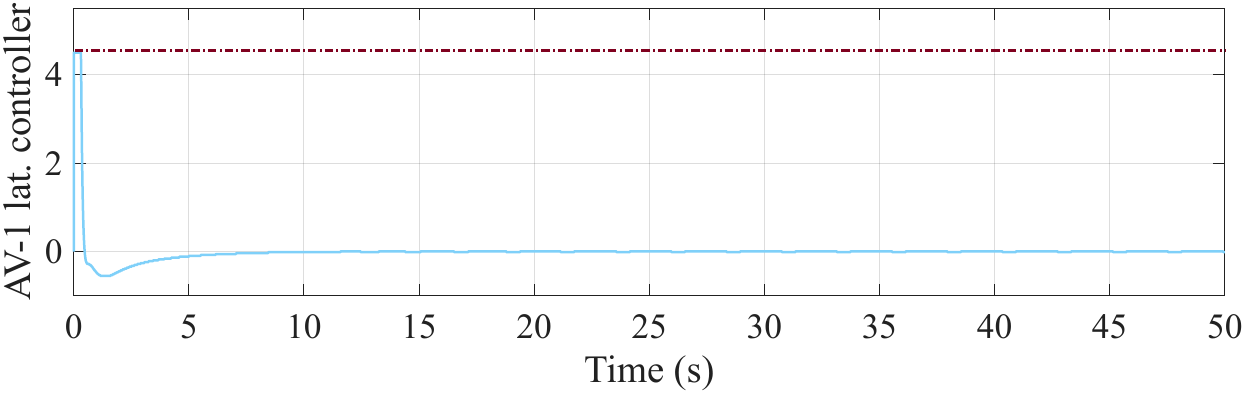}\\
            \includegraphics[width=0.475\linewidth]{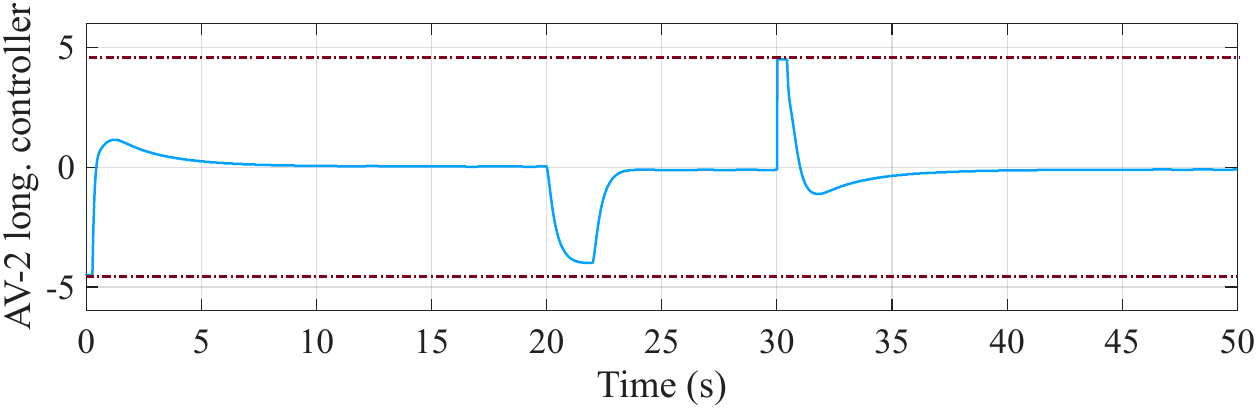}
            \includegraphics[width=0.475\linewidth]{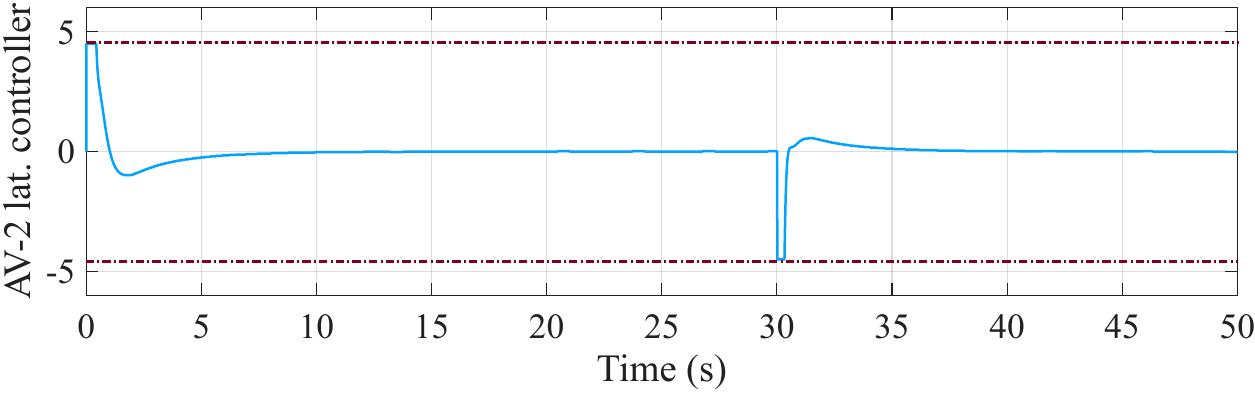}\\
            \includegraphics[width=0.475\linewidth]{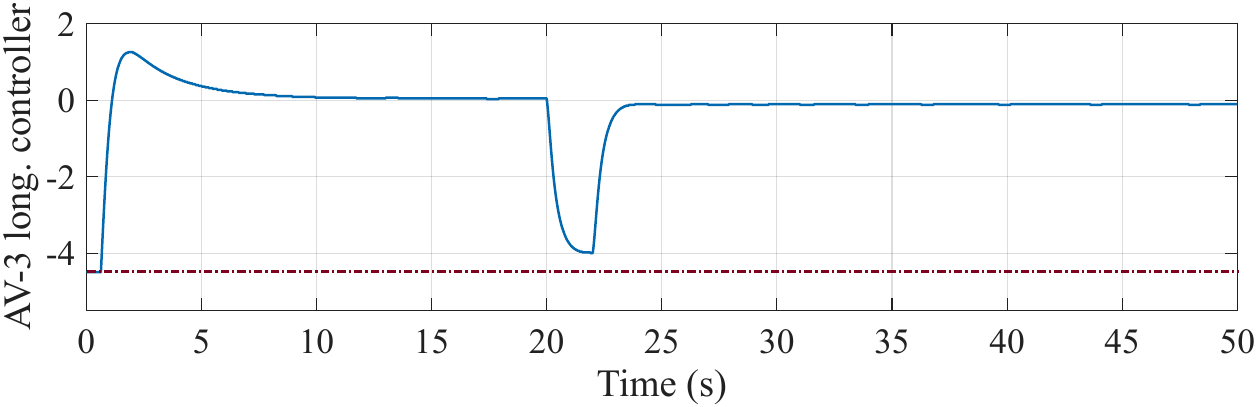}
            \includegraphics[width=0.475\linewidth]{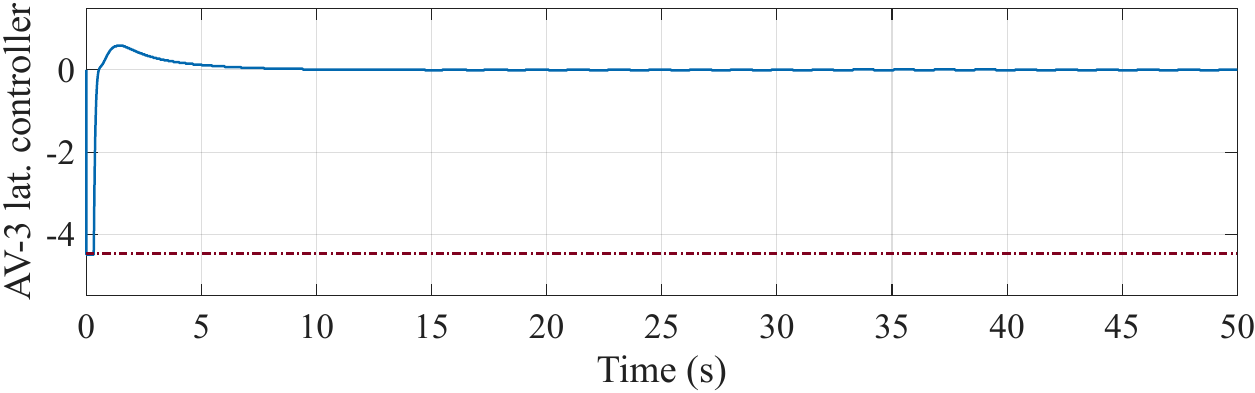}\\
            \includegraphics[width=0.475\linewidth]{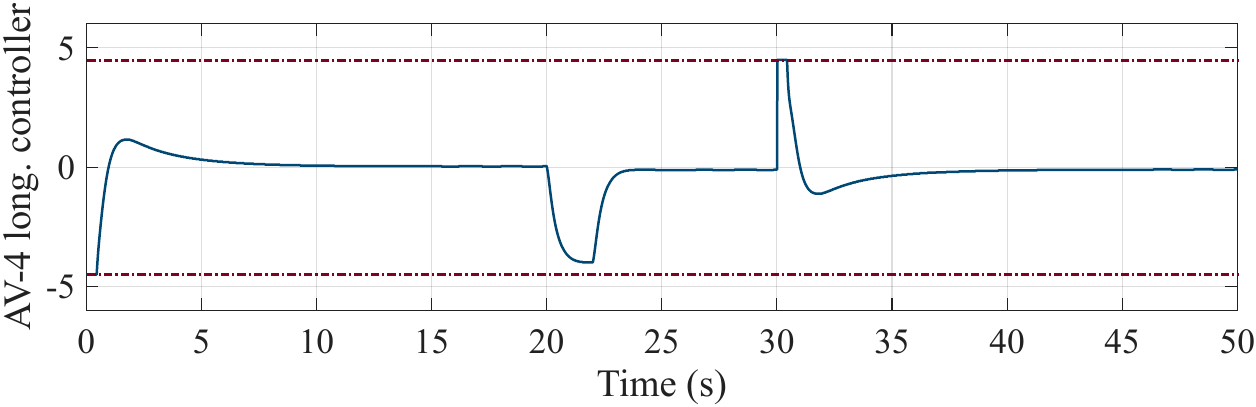}
            \includegraphics[width=0.475\linewidth]{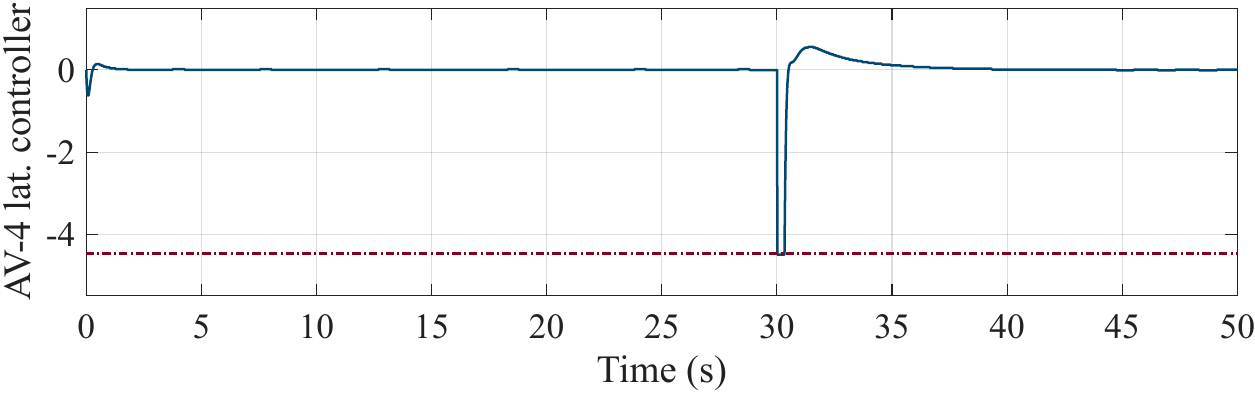}
            \caption{The longitudinal and lateral input controllers for the AV-$i$ with input saturation.}
            \label{controller}
        \end{figure*}

\subsection{Analysis on Control Performance}
The four AVs cooperative lane changing formation control scenario shown in Fig.~\ref{running} is tested. 
We implemented the system prototype using Godot 4.5.2 Mono (\url{https://godotengine.org}). The system takes simulated signals as input and presents the results through an interactive visualization interface. Fig.~\ref{3D} presents selected screenshots of the 3D visualization of the AV fleet. The full video can be accessed via the following link (\url{https://youtu.be/LZdUl_v4jvw}). It is worth noting that Fig.~\ref{3D} here shows the dynamic changes in the structure depicted in Fig.~\ref{running}, and the video provides a better illustration of the evolution of the vehicle fleet. 
Fig.~\ref{tracking} illustrates the longitudinal and lateral tracking control performance of AV-$i$. 
In the longitudinal direction, the four AVs rapidly reach their respective desired positions and accomplish the tracking task at the initial stage. 
In addition, when the velocity changes at 20 s, all vehicles exhibit smooth and coordinated trajectory adjustments. When the fleet formation changes at 30 s, the platoon completes the longitudinal transition to a side by side arrangement in pairs within approximately 5 s. 
In the lateral direction, the four AVs require about 12 s at the initial stage to complete the lane-merging maneuver from four separate lanes into a single lane. Similarly, when the fleet configuration changes at 30 s, the platoon takes approximately 10 s to accomplish the lane-change process in which AV-$2$ and AV-$4$ move to the adjacent lane. Overall, the entire tracking control process remains smooth and stable.

Fig.~\ref{velocity} presents the longitudinal and lateral velocities of AV-$i$ under the relative-threshold event-triggered control strategy. In the longitudinal direction, because each vehicle has a different initial velocity, speed fluctuations appear at the beginning. However, all four AVs quickly converge to the desired velocity. When the desired velocity changes at 20 s, slight delays are observed in the responses of the four vehicles, yet they all decelerate in a coordinated manner. At 30 s, AV-$2$ and AV-$4$ accelerate to accomplish the lane-change maneuver. In the lateral direction, when the four vehicles initially perform lane changes, their lateral velocities vary to different extents, and then quickly converge to 0 m/s after the lane-changing process is completed. Similarly, at 30 s, fluctuations in lateral velocity appear for AV-$2$ and AV-$4$ during the lane-change maneuver.

The longitudinal and lateral control inputs of AV-$i$ with input saturation are shown in Fig. \ref{controller}. It can be observed that, during the control process, the control signals repeatedly exhibit tendencies to exceed the prescribed limits, which may cause damage to the torque-related actuating components of the vehicles in practical scenarios. After the input saturation constraint is taken into account, the control signals remain confined within the predefined range, thereby preventing actuator saturation from destabilizing the closed-loop system. 

Table \ref{triggercounts} and Fig. \ref{trigger} illustrate the results of the ETC strategy. In Table \ref{triggercounts}, the rate is calculated as
\begin{equation}
		\nonumber
		\begin{aligned}
			\text{Rate}=\frac{\text{total counts-trigger counts}}{\text{total counts}}\times 100\%,
		\end{aligned}
	\end{equation}
which represents the proportion of control resources saved by the ETC strategy. A higher value indicates greater control efficiency. Compared with \cite{2025}, the dynamic-threshold ETC strategy adopted in this paper achieves a higher level of resource-saving efficiency. In addition, Fig. \ref{trigger} presents the triggering profiles in both the longitudinal and lateral directions, thereby providing an intuitive visualization of the discrete updating behavior of the AVs' controller signals.

In Fig.~\ref{safe}, the inter-vehicle safe distances throughout the control process are presented, indicating that the proposed dynamic-threshold ETC-based formation control strategy consistently maintains safe vehicle spacing and thereby guarantees the safety performance of the overall traffic system.

\textbf{Implications for traffic flow.} 
From the perspective of traffic flow, the above simulation results further indicate that the proposed controller can support coordinated and stable fleet operation during speed variations and formation reconfigurations. 
For example, during the 20–22 s deceleration in Fig.~\ref{velocity}, the longitudinal velocities of all four vehicles converge rapidly without overshoot, and the inter‑vehicle distances in Fig.~\ref{safe} remain consistently above the safe threshold. 
This indicates negligible disturbance propagation and stable traffic flow.
The smooth longitudinal responses in Fig.~\ref{tracking} and Fig.~\ref{velocity} imply that the vehicles can complete acceleration and deceleration without inducing excessive speed oscillations, which is important for suppressing disturbance propagation along the fleet. Meanwhile, the bounded lateral maneuvers enable the AVs to merge and reconfigure across lanes while maintaining ordered vehicle motion. The safe distance profiles in Fig.~\ref{safe} further show that the proposed scheme preserves safe inter vehicle spacing throughout the entire maneuver. Therefore, the proposed formation control strategy is beneficial for improving traffic flow regularity, lane utilization, and operational safety in multi lane autonomous driving scenarios, while the event triggered implementation reduces communication and controller update demands.

 \begin{figure*}[!t]
    \centering
    \footnotesize
    \captionof{table}{Statistics of trigger counts (where\textbf{ Rate}$^{*}$ denotes the ETC results in \cite{2025})}
    \label{triggercounts}
    \begin{tabular}{ccccccc}
        \toprule
        \textbf{ } & \multicolumn{2}{c}{\textbf{Continuous-time Controller}} &
			\multicolumn{2}{c}{\textbf{Event-triggered Controller}} & 
			 \textbf{ }  & \textbf{ } \\
            \cmidrule(lr){2-3} \cmidrule(lr){4-5}
        \textbf{No.} & \textbf{Longitudinal} & \textbf{Lateral} & \textbf{Longitudinal} & \textbf{Lateral} & \textbf{Rate} & \textbf{Rate$^{*}$} \\
        \midrule
        AV1 & 50000 & 50000 & 358 & 127 & \textbf{95.15\%} $\uparrow$ & 85.93\% \\
        AV2 & 50000 & 50000 & 517 & 297 & \textbf{91.86\%} $\uparrow$ & 51.37\% \\
        AV3 & 50000 & 50000 & 419 & 148 & \textbf{94.33\%} $\uparrow$ & 42.35\% \\
        AV4 & 50000 & 50000 & 576 & 153 & \textbf{92.71\%} $\uparrow$ & 33.17\% \\
        \bottomrule
    \end{tabular}
    \end{figure*}

          \begin{figure*}
            \centering
            \includegraphics[width=0.475\linewidth]{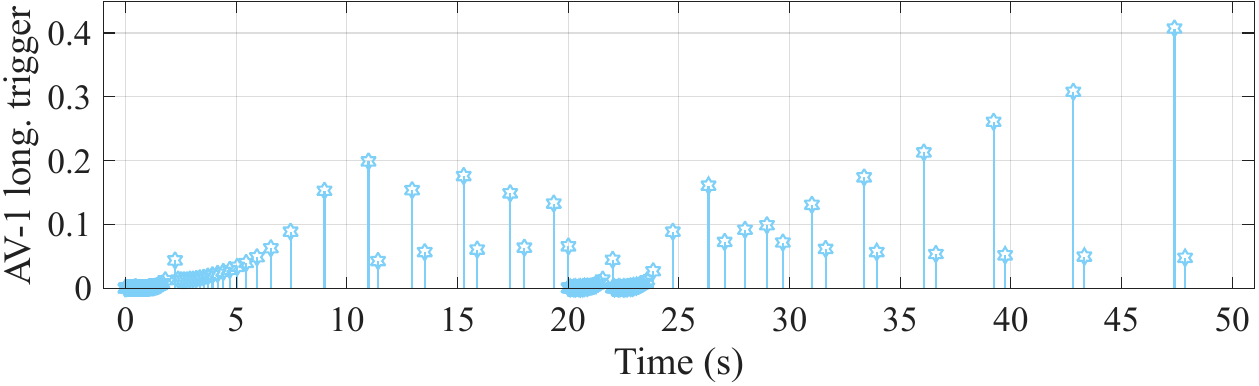}
            \includegraphics[width=0.475\linewidth]{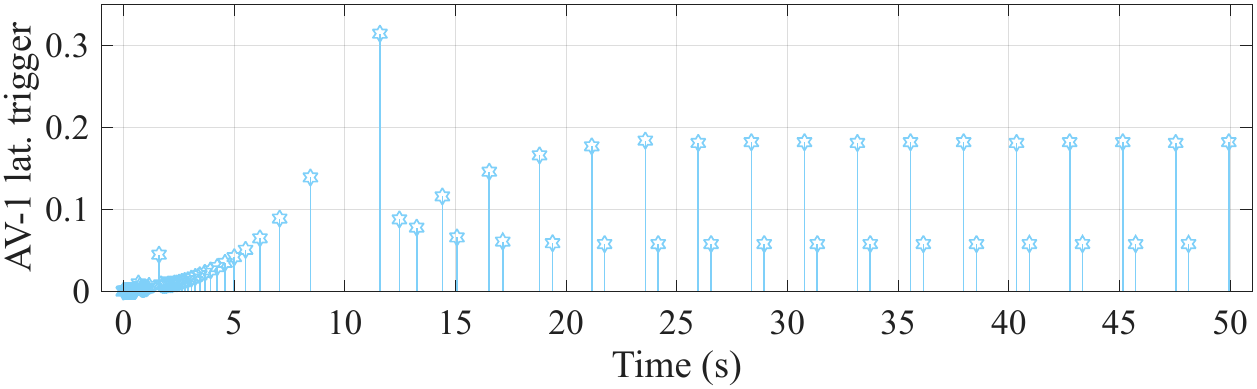}\\
            \includegraphics[width=0.475\linewidth]{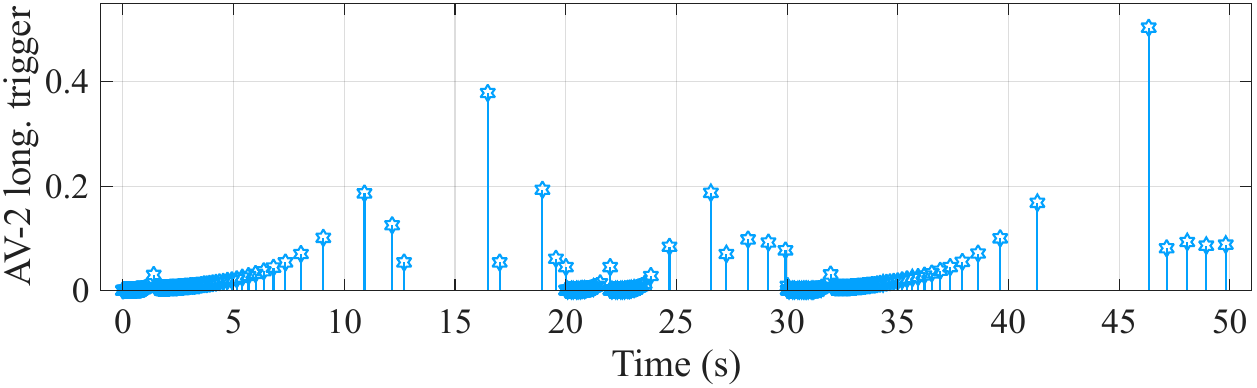}
            \includegraphics[width=0.475\linewidth]{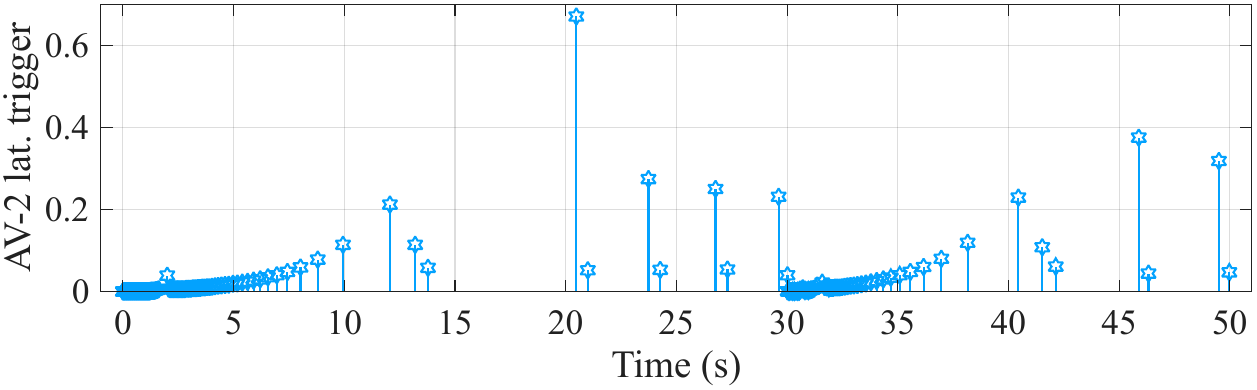}\\
            \includegraphics[width=0.475\linewidth]{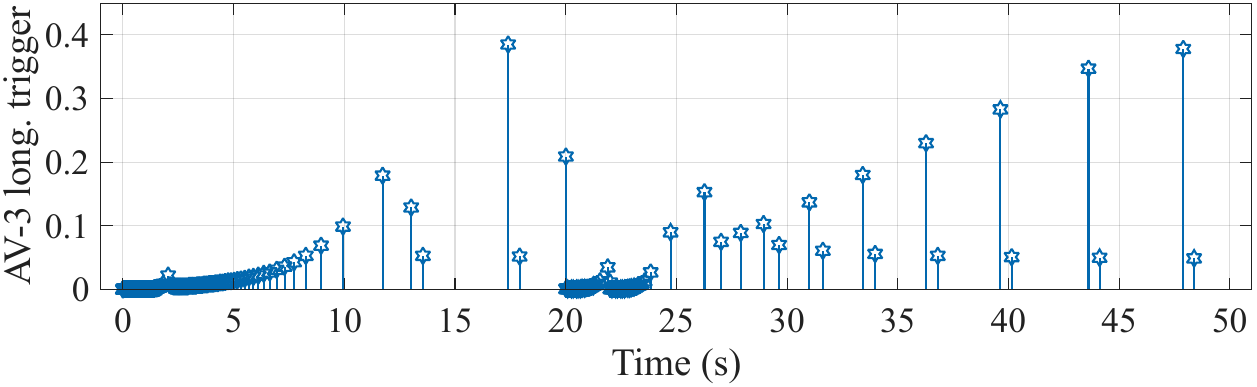}
            \includegraphics[width=0.475\linewidth]{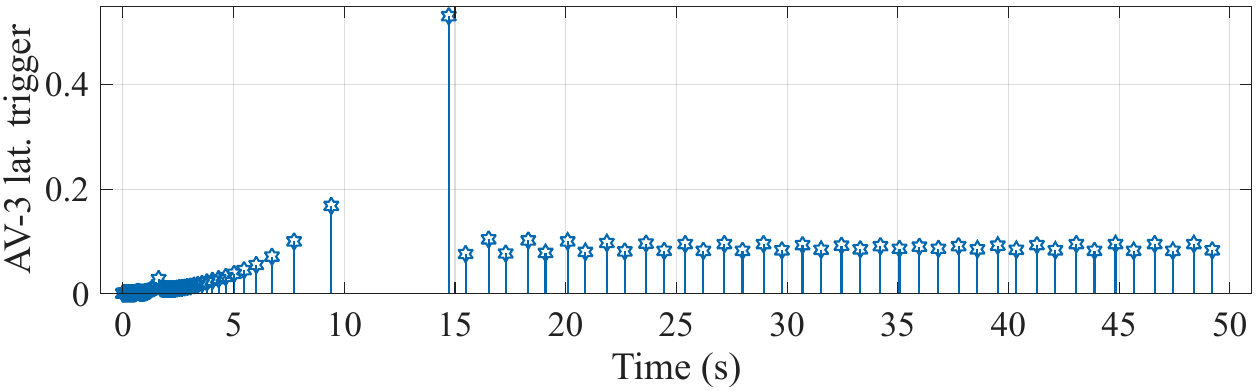}\\
            \includegraphics[width=0.475\linewidth]{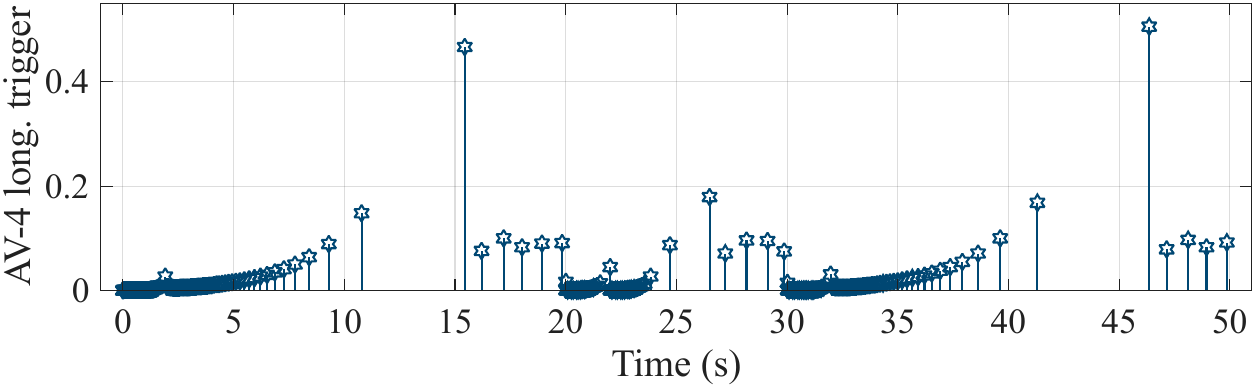}
            \includegraphics[width=0.475\linewidth]{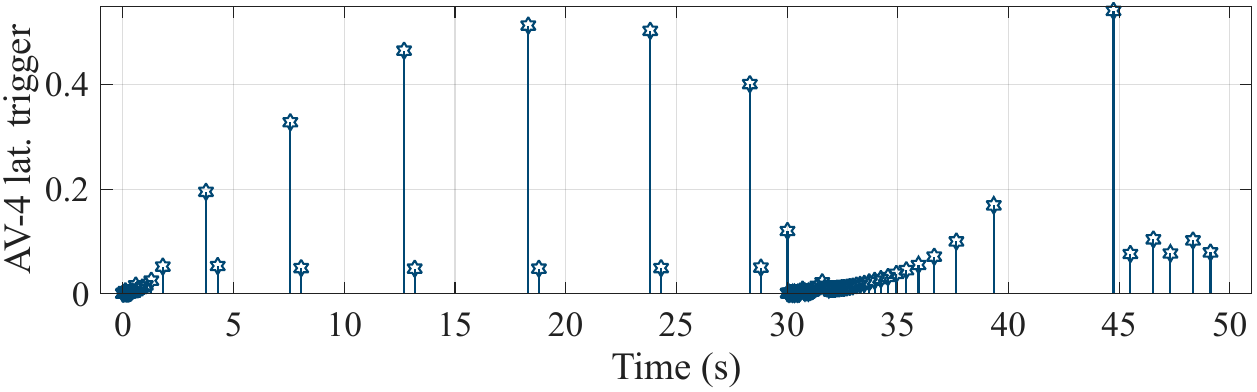}
            \caption{The longitudinal and lateral update time intervals of control laws for the AV-$i$.}
            \label{trigger}
        \end{figure*}

        \begin{figure}
            \centering
            \includegraphics[width=0.49\textwidth]{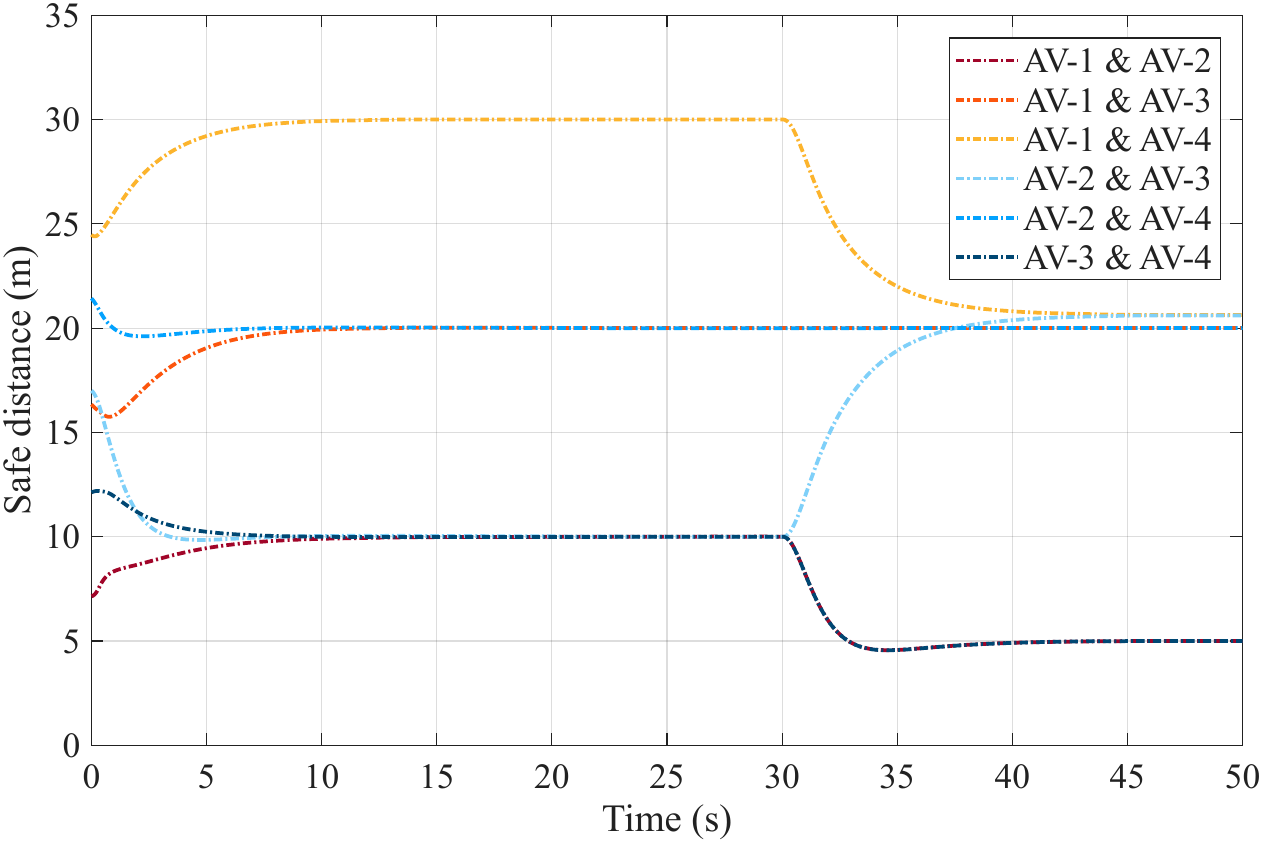}
            \caption{The safe distances between any two AVs.}
            \label{safe}
          \end{figure}


    \subsection{Comparisons with Existing Literature}

    We now compare our proposed collaborative adaptive formation control framework with the existing methods in \cite{lemma3, 2025}. 
    These studies all investigate consensus tracking problems in traffic scenarios involving four vehicles moving from different initial states to prescribed destinations, while employing neural networks and considering external disturbances. 
    However, unlike \cite{lemma3, 2025}, our proposed framework explicitly incorporates input saturation to reflect the physical limitations of throttle, braking, and steering, as shown in Fig.~\ref{controller}. 
    In addition, to better capture practical operating conditions, including perception delay, computation latency, and actuator lag, a time delay is included in the system model in \eqref{equ1}. 
    These factors make the control problem considerably more challenging. 
    Compared with \cite{2025}, the formation task considered in this paper is also more complex, and the desired velocity evolves from a constant-speed profile to a fully time-varying one. 
    Furthermore, unlike \cite{intro3, intro4, intro5, intro6, intro7, intro8, jiang}, which mainly present point-mass simulations, this paper provides a 3D visualization of the AV fleet using the simulated motion data, allowing a more intuitive assessment of the formation reconfiguration process and of practical safety issues such as vehicle overlap and collision risk.

	\section{Conclusion}\label{sec6}
	
    This paper investigates the collaborative adaptive formation control framework for AVs subject to input delay and saturation constraints, and develops a backstepping-based control framework with enhanced practical applicability. By incorporating command filters and observer-based techniques, an adaptive neural-network tracking controller together with a dynamic-threshold ETC strategy is constructed. The proposed method effectively addresses delay compensation, saturation nonlinearity, and computational burden reduction within a unified framework. Meanwhile, accurate trajectory tracking is achieved through synchronized sampling, and the convergence of velocity and position tracking errors is rigorously guaranteed. Simulation results further demonstrate that the developed approach can substantially reduce the control update frequency while preserving satisfactory tracking performance and coordinated fleet behavior. Therefore, the proposed method provides an effective and practically feasible solution for the cooperative control of AVs.

    Future work will focus on extending the proposed framework to more complicated traffic scenarios and dynamically varying fleet configurations, so as to further improve the efficiency, robustness, and adaptability of AV control systems in realistic environments.


    \flushend

    \begin{IEEEbiography}[{\includegraphics[width=1in,height=1.25in,clip,keepaspectratio]{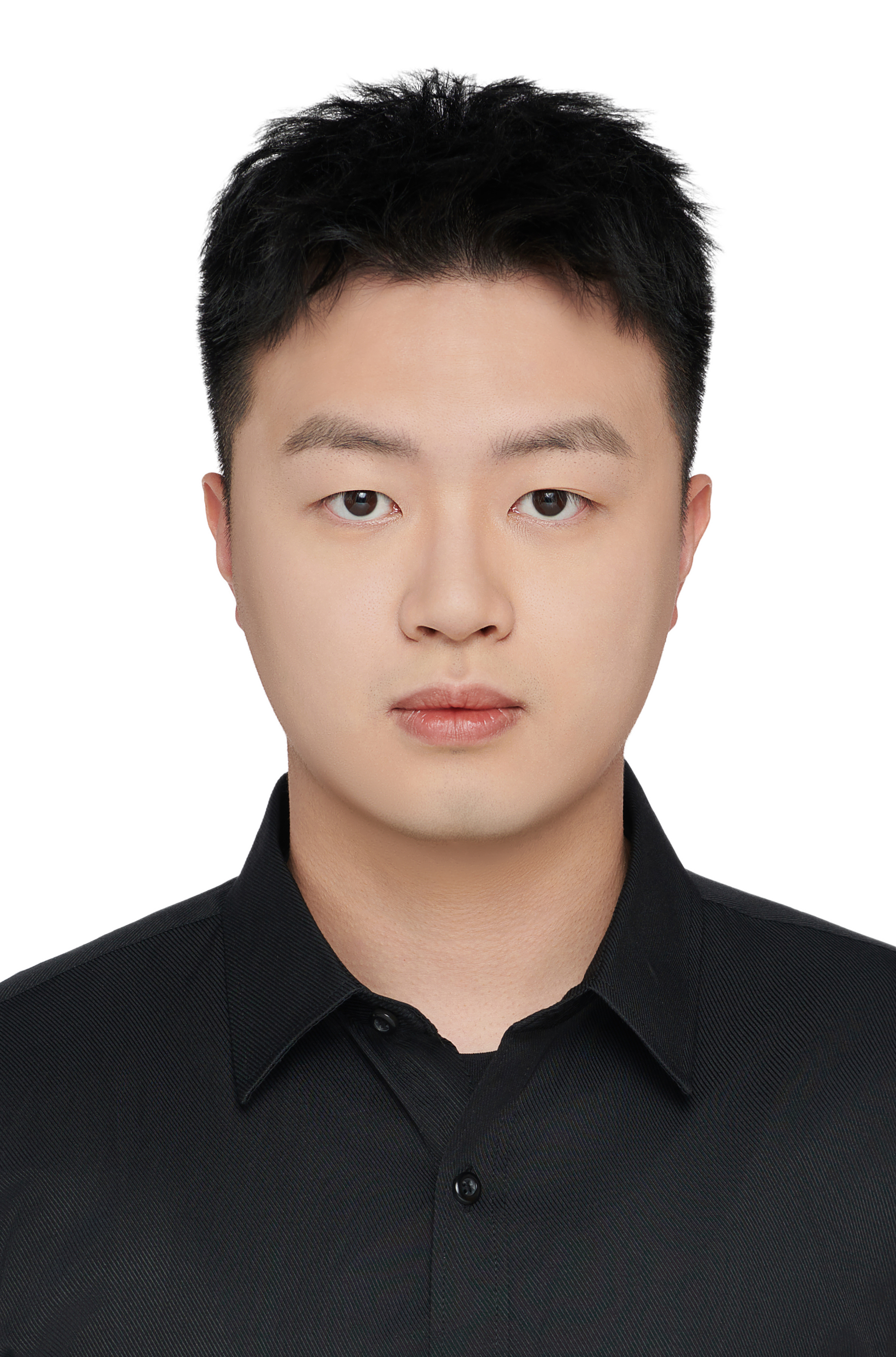}}]{Ziming Wang} received the B.Eng. degree in Electronic Information Engineering from Southwest University in 2023 and the M.Phil. degree in the Trust of Robotics and Autonomous Systems from the Hong Kong University of Science and Technology (Guangzhou) in 2025. He was a visiting scholar at the School of Vehicle and Mobility, Tsinghua University. He is currently a PhD student in the Trust of Artificial Intelligence from the HKUST-GZ. He is a senior reviewer of the IEEE TITS, IEEE TSMC, IEEE TMC and IEEE/CAA JAS. His research interests include control theory, optimization, multi-agent systems, reinforcement learning, and their applications in intelligent transportation systems. 
    \end{IEEEbiography}

    \begin{IEEEbiography}[{\includegraphics[width=1in,height=1.25in,clip,keepaspectratio]{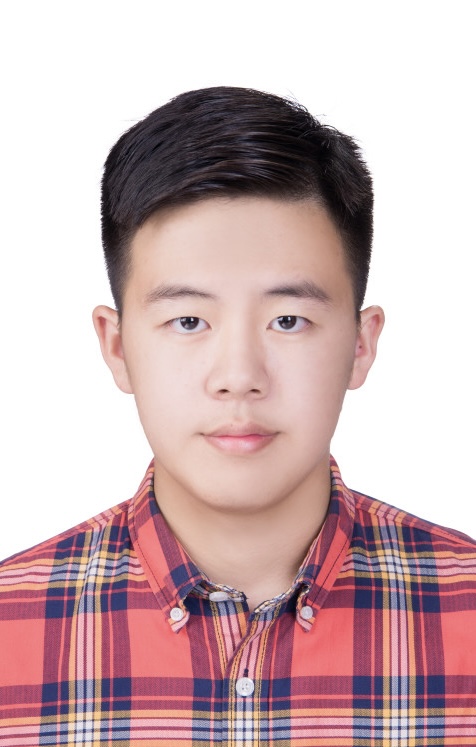}}]{Guanxuan Jiang} received the B.S. degree in Computer Science from the University of Sheffield, which included a one-year study-abroad program at the Chinese University of Hong Kong. He obtained the MPhil degree in the Trust of Computational Media and Arts from the Hong Kong University of Science and Technology (Guangzhou), where he is currently pursuing his PhD. His research interests include computational intelligence, ubiquitous computing, data-driven interaction, and extended reality.
    \end{IEEEbiography}


        

       \begin{IEEEbiography}[{\includegraphics[width=1in,height=1.25in,clip,keepaspectratio]{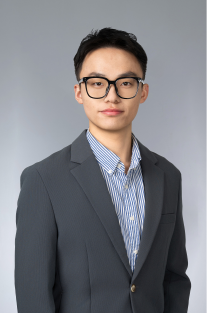}}]{Yihuai Zhang} received his B.E. degree in Vehicle Engineering from Southwest University in 2019, and his M.S. degree in Vehicle Engineering from South China University of Technology in 2022, and his Ph.D. degree in Intelligent Transportation at the Hong Kong University of Science and Technology (Guangzhou) in 2026. His research focuses on distributed parameter systems, learning and control for dynamical systems, and their applications in intelligent transportation systems.
       \end{IEEEbiography}

    \begin{IEEEbiography}[{\includegraphics[width=1in,height=1.25in,clip,keepaspectratio]{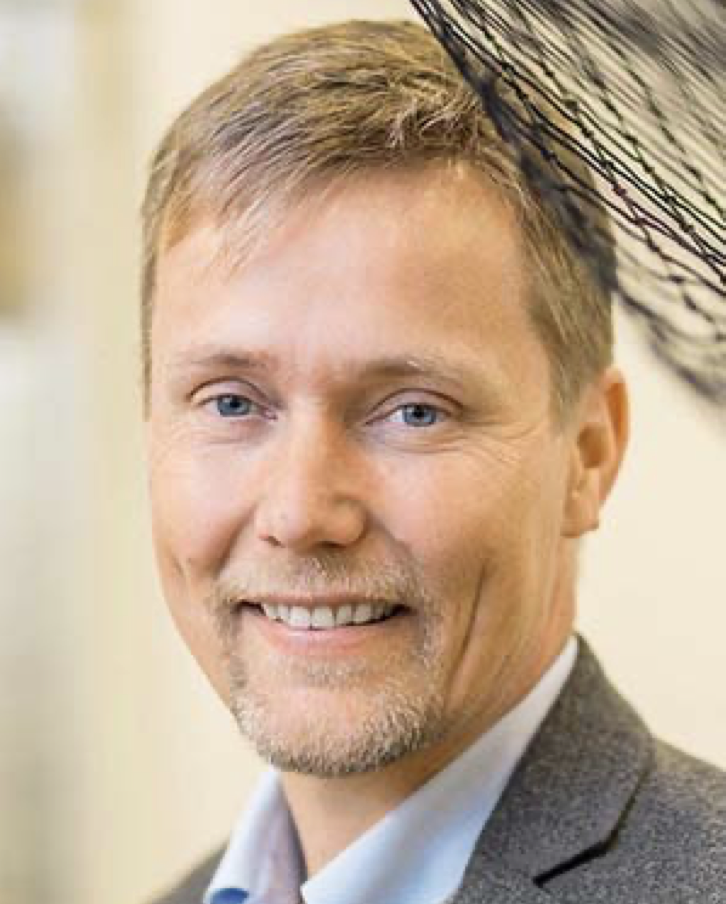}}]{Karl Henrik Johansson} (Fellow, IEEE) is Swedish Research Council Distinguished Professor in Electrical Engineering and Computer Science at KTH Royal Institute of Technology in Sweden and Founding Director of Digital Futures. He earned his MSc degree in Electrical Engineering and PhD in Automatic Control from Lund University. He has held visiting positions at UC Berkeley, Caltech, NTU and other prestigious institutions. His research interests focus on networked control systems and cyber-physical systems with applications in transportation, energy, and automation networks. For his scientific contributions, he has received numerous best paper awards and various distinctions from IEEE, IFAC, and other organizations. He has been awarded Distinguished Professor by the Swedish Research Council, Wallenberg Scholar by the Knut and Alice Wallenberg Foundation, Future Research Leader by the Swedish Foundation for Strategic Research. He has also received the triennial IFAC Young Author Prize and IEEE CSS Distinguished Lecturer. He is the recipient of the 2024 IEEE CSS Hendrik W. Bode Lecture Prize. His extensive service to the academic community includes being President of the European Control Association, IEEE CSS Vice President Diversity, Outreach $\&$ Development, and Member of IEEE CSS Board of Governors and IFAC Council. He has served on the editorial boards of Automatica, IEEE TAC, IEEE TCNS and many other journals. He has also been a member of the Swedish Scientific Council for Natural Sciences and Engineering Sciences. He is Fellow of both the IEEE and the Royal Swedish Academy of Engineering Sciences.
    \end{IEEEbiography}

    \begin{IEEEbiography}[{\includegraphics[width=1in,height=1.25in,clip,keepaspectratio]{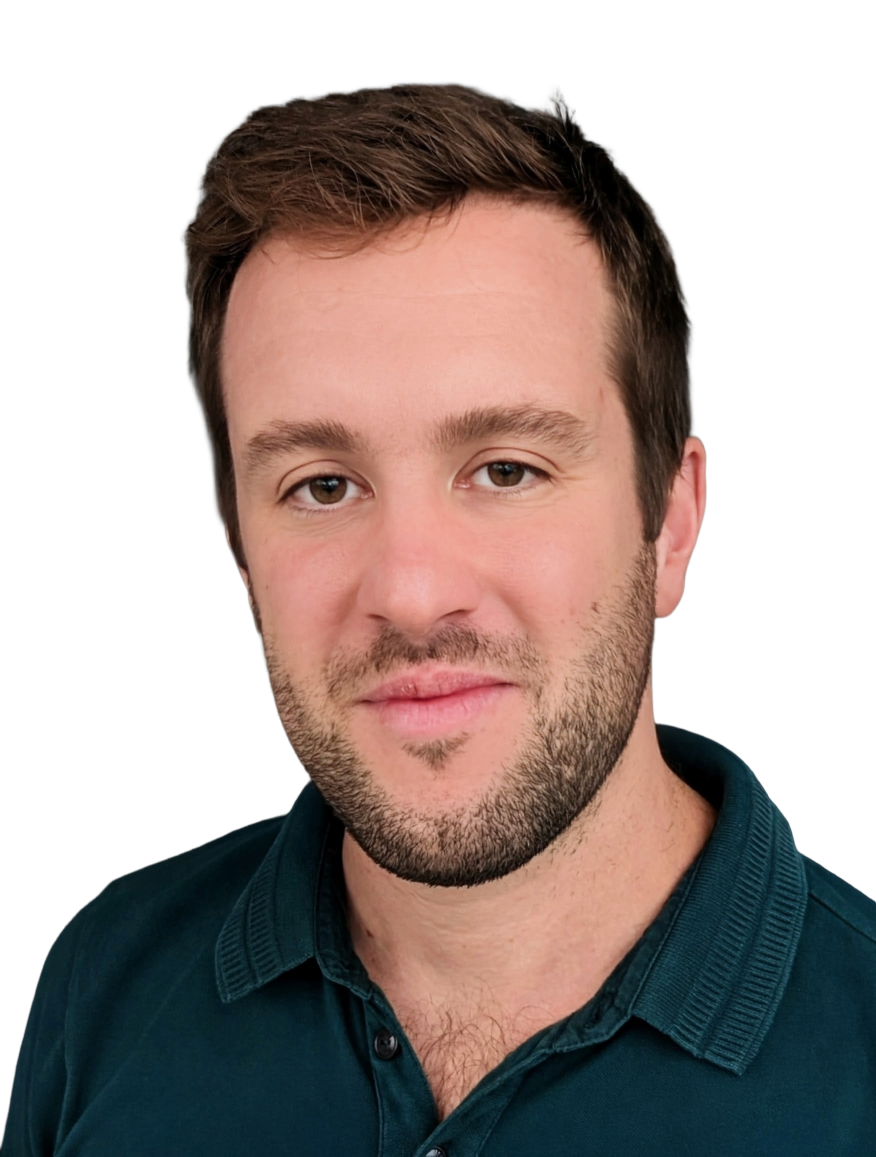}}]{Apostolos I. Rikos} (Member, IEEE) is an Assistant Professor at the Artificial Intelligence Thrust of the Information Hub, The Hong Kong University of Science and Technology (Guangzhou), Guangzhou, China. 
    He is also affiliated with the Department of Computer Science and Engineering, The Hong Kong University of Science and Technology, Clear Water Bay, Hong Kong, China.
    He received his B.Sc., M.Sc., and Ph.D. degrees in Electrical Engineering from the Department of Electrical and Computer Engineering, University of Cyprus in 2010, 2012, and 2018, respectively.
    In 2018, he joined the KIOS Research and Innovation Center of Excellence in Cyprus, where he was a Research Lecturer. 
    In 2020, he joined the Division of Decision and Control Systems of KTH Royal Institute of Technology as a Postdoctoral Researcher. 
    In 2023, he joined the Department of Electrical and Computer Engineering, Division of Systems Engineering, at Boston University as a Postdoctoral Associate.
    His research interests are in the area of distributed optimization and learning, distributed network control and coordination, privacy and security. 
    \end{IEEEbiography}

\end{document}